\documentclass[a4paper,fleqn]{cas-dc}

\usepackage[numbers]{natbib} 

\usepackage{caption}
\usepackage{multirow}
\usepackage{longtable}
\usepackage{tabu}
\usepackage{colortbl,booktabs}
\usepackage{amsmath}
\usepackage{amsthm,amssymb}
\usepackage{amsfonts}
\usepackage{diagbox}
\usepackage{multicol}
\usepackage{url}
\usepackage{graphicx}

\usepackage{verbatim}
\usepackage{algorithm}
\usepackage{algorithmicx}

\usepackage{soul,color,xcolor}

\usepackage[noend]{algpseudocode}
\makeatletter
\newenvironment{breakablealgorithm}
  {
   \begin{center}
     \refstepcounter{algorithm}
     \hrule height.8pt depth0pt \kern2pt
     \renewcommand{\caption}[2][\relax]{
       {\raggedright\textbf{\ALG@name~\thealgorithm} ##2\par}%
       \ifx\relax##1\relax 
         \addcontentsline{loa}{algorithm}{\protect\numberline{\thealgorithm}##2}%
       \else 
         \addcontentsline{loa}{algorithm}{\protect\numberline{\thealgorithm}##1}%
       \fi
       \kern2pt\hrule\kern2pt
     }
  }{
     \kern2pt\hrule height.8pt depth0pt \kern2pt\relax
   \end{center}
  }
\makeatother
\errorcontextlines\maxdimen
\makeatletter
    \newcommand*{\algrule}[1][\algorithmicindent]{\makebox[#1][l]{\hspace*{.5em}\thealgruleextra\vrule height \thealgruleheight depth \thealgruledepth}}%
\newcommand*{\thealgruleextra}{}
\newcommand*{\thealgruleheight}{.75\baselineskip}
\newcommand*{\thealgruledepth}{.25\baselineskip}
\newcount\ALG@printindent@tempcnta
\def\ALG@printindent{%
    \ifnum \theALG@nested>0
        \ifx\ALG@text\ALG@x@notext
        \else
            \unskip
            \addvspace{-1pt}
            \ALG@printindent@tempcnta=1
            \loop
                \algrule[\csname ALG@ind@\the\ALG@printindent@tempcnta\endcsname]%
                \advance \ALG@printindent@tempcnta 1
            \ifnum \ALG@printindent@tempcnta<\numexpr\theALG@nested+1\relax
            \repeat
        \fi
    \fi
    }%
\usepackage{etoolbox}
\patchcmd{\ALG@doentity}{\noindent\hskip\ALG@tlm}{\ALG@printindent}{}{\errmessage{failed to patch}}
\makeatother
\newbox\statebox
\newcommand{\myState}[1]{%
    \setbox\statebox=\vbox{#1}%
    \edef\thealgruleheight{\dimexpr \the\ht\statebox+1pt\relax}%
    \edef\thealgruledepth{\dimexpr \the\dp\statebox+1pt\relax}%
    \ifdim\thealgruleheight<.75\baselineskip
        \def\thealgruleheight{\dimexpr .75\baselineskip+1pt\relax}%
    \fi
    \ifdim\thealgruledepth<.25\baselineskip
        \def\thealgruledepth{\dimexpr .25\baselineskip+1pt\relax}%
    \fi
    \State #1%
    \def\thealgruleheight{\dimexpr .75\baselineskip+1pt\relax}%
    \def\thealgruledepth{\dimexpr .25\baselineskip+1pt\relax}%
}

\newtheorem{definition}{Definition}
\newtheorem{theorem}{Theorem}

\newtheorem{remark}{Remark}%

\def\tsc#1{\csdef{#1}{\textsc{\lowercase{#1}}\xspace}}
\tsc{WGM}
\tsc{QE}

\begin{document}
\let\WriteBookmarks\relax
\def\floatpagepagefraction{1}
\def\textpagefraction{.001}

\shorttitle{Program Dependence Net and On-demand Slicing for Property Verification of Concurrent System and Software}    

\shortauthors{Zhijun Ding, Shuo Li, Cheng Chen et al.}  

\title [mode = title]{Program Dependence Net and On-demand Slicing for Property Verification of Concurrent System and Software}  



%

\author[1]{Zhijun Ding}[orcid=0000-0003-2178-6201]
\ead{dingzj@tongji.edu.cn}
\credit{Conceptualization, Writing - Review \& Editing}
\affiliation[1]{organization={Tongji University},
            city={Shanghai},
            postcode={201804}, 
            country={China}}

\author[1]{Shuo Li}[orcid=0000-0002-1984-6271]
\ead{lishuo20062002@126.com}
\cormark[1]
\credit{Methodology, Formal analysis, Writing - Original Draft}
\cortext[1]{Corresponding author}

\author[1]{Cheng Chen}
\credit{Software}

\author[1]{Cong He}
\credit{Software}


\begin{abstract}
When checking concurrent software using a finite-state model, we face a formidable state explosion problem. One solution to this problem is dependence-based program slicing, whose use can effectively reduce verification time. It is orthogonal to other model-checking reduction techniques. However, when slicing concurrent programs for model checking, there are conversions between multiple irreplaceable models, and dependencies need to be found for variables irrelevant to the verified property, which results in redundant computation. To resolve this issue, we propose a Program Dependence Net (PDNet) based on Petri net theory. It is a unified model that combines a control-flow structure with dependencies to avoid conversions. For reduction, we present a PDNet slicing method to capture the relevant variables' dependencies when needed. PDNet in verifying linear temporal logic and its on-demand slicing can be used to significantly reduce computation cost. We implement a model-checking tool based on PDNet and its on-demand slicing, and validate the advantages of our proposed methods.
\end{abstract}




\maketitle

\section{Introduction}\label{Sec:Int}

Verifying linear temporal logic (LTL) properties for concurrent systems and software is a challenging task. Finite-state model checking \cite{Clarke2018Handbook} is one of the most widely used methods for this purpose. However, the state explosion problem seriously hinders its practical application. To address this issue, researchers have developed reduction techniques like partial order \cite{Albert2018Constrained} from a state-space perspective. This technique can reduce possible orderings of independent statements, but it cannot guarantee that all orderings irrelevant to the verified property are completely reduced.

The program slicing theory suggests that a slicing criterion can be used to identify the essential parts needed to capture all relevant information for verified properties. Slicing methods with their focus on properties have been a widely used reduction technique in software verification \cite{Hatcliff1999A,Hatcliff2001Using,Chalupa2021Symbiotic}.
The existing evaluations \cite{Dwyer2006Evaluating,Chalupa2019Evaluation,Kumar2021Program} have confirmed that they can be effective in reducing the verification time and be orthogonal to other reduction techniques for model checking.
Traditionally, program slicing builds a program dependence graph (PDG) \cite{Ferrante1987The,Horwitz1990Interprocedural,Qi2017Precise} based on a control-flow graph (CFG).
In PDG, nodes represent statements and edges capture the relationships among them. These relationships fall into two categories: control-flow and data-flow dependencies. The former is determined by analyzing the conditions that dictate statement execution in CFG, while the latter is established by the definition-use relationship of variables on CFG's nodes. By applying transitive closure to PDG's edges starting from the nodes meeting a slicing criterion, one can identify the remaining nodes. These nodes correspond to the statements in the residual program, also known as a program slice.
Dependence-based program slicing is commonly employed as a preprocessing technique for model checking \cite{Hatcliff1999A,Hatcliff2000Slicing}. Control-flow automata (CFAs) are used to represent a slice, which are formal models that describe a control-flow structure, i.e., the order in which statements are executed, in many advanced tools \cite{Lowe2014CPAchecker}.
CFA's nodes signify control locations, with the edges denoting program operations. The executed operation is labeled on the edges between two nodes when a control location transitions from a source to destination. The reachable tree of CFA is utilized to determine state space for a model checking purpose.

The above-mentioned methods have the drawback of imposing significant computation cost. Firstly, the utilization of multiple models (such as PDG for slicing and CFA for model checking) at distinct stages necessitates conversions between them. The implementation of a unified model can obviate the necessity for such conversions and diminish computation cost from PDG to CFA.
Secondly, to calculate PDG, all data-flow dependencies must be thoroughly captured in advance. It is important to note that some variables' data-flow dependencies may not be relevant to the verified property in PDG.
Our proposal aims to tackle these two shortcomings by presenting a unified model that combines a control-flow structure with program dependencies. This model can determine the necessary data-flow dependencies as needed during slicing instead of capturing them beforehand. However, implementing this model on CFA poses a challenge.
CFA edges represent executed operations or statements, while PDG edges indicate program dependencies among operations or statements. Thus, both types of edges are crucial and irreplaceable, making it difficult to merge them directly and determine data-flow dependencies based on CFA solely.

As a concurrent system model, Petri nets (PNs)
\cite{Wolf2019Petri,Abubakar2023Adaptive,Cheng2023Optimal,Luo2021Inference,Luo2019Robust,Liu2017Deadlock,You2017Computation} can represent the control-flow structure of a concurrent program. An automaton-theoretic \cite{Kahlon2006An} approach can be applied to PNs for LTL model checking.
In the context of CFA, an edge is limited to a single source. Differently, a PN's transition's input place is not constrained to a single source. It allows for the representation of a statement's execution order condition and domination condition via distinct input places of the corresponding transition. PNs provide an appropriate means to amalgamate a control-flow structure and program dependencies.
Furthermore, the definition-use relationship of variables used to capture data-flow dependencies should be calculated by using variable operations and a control-flow structure.
Colored Petri nets (CPNs), which are a form of high-level Petri net, can represent variable operations with colored places and expressions, thereby enabling PNs to effectively capture data-flow dependencies as needed.

While current verification methods employing PNs or CPNs exist \cite{Gan2018An,Dietsch2021Verification}, they fail to fully combine a control-flow structure with control-flow dependencies.
As a solution, we propose a new model, \textit{Program Dependence Net} (PDNet), which leverages CPNs as a unified model to minimize the computation cost incurred in above-mentioned conversions.
Some methods of PN slicing \cite{Khan2018PNSlicing} are used to reduce concurrent models like workflows. However, they fail to take program dependencies into consideration. To address the mentioned issues, we propose on-demand slicing of PDNet to reduce computation cost associated with data-flow dependencies of irrelevant variables.
We aim to make the following novel contributions to the field of model checking:

1. We propose PDNet as a unified model for concurrent programs. It combines a control-flow structure with control-flow dependencies, thus avoiding the computation cost of model conversions required by traditional PDG-based program slicing methods.

2. We propose a new method to reduce PDNet by using a slicing technique that extracts criterion from LTL formulae. The key to our approach is capturing data-flow dependencies with an on-demand way based on PDNet, thus avoiding unnecessary computation.

3. We implement a model checking tool named \textit{DAMER}, standing for Dependence Analyser and Multi-threaded programs checkER. It automatically translates concurrent programs to PDNet without any manual intervention and reduces PDNet by on-demand slicing.

The next section gives a motivating examples.
Section \ref{Sec:PDNet} describes PDNet. Section \ref{Sec:Model} presents dependency modeling with PDNet. Section \ref{Sec:Slice} introduces on-demand PDNet slicing. Section \ref{Sec:Exp} discusses experimental results. Section \ref{Sec:Rel} briefs the related works. Section \ref{Sec:Con} concludes this article.

\section{Motivating Example}\label{Sec:Exm}

\begin{figure*}[!h]\centering
 \includegraphics[width=\textwidth]{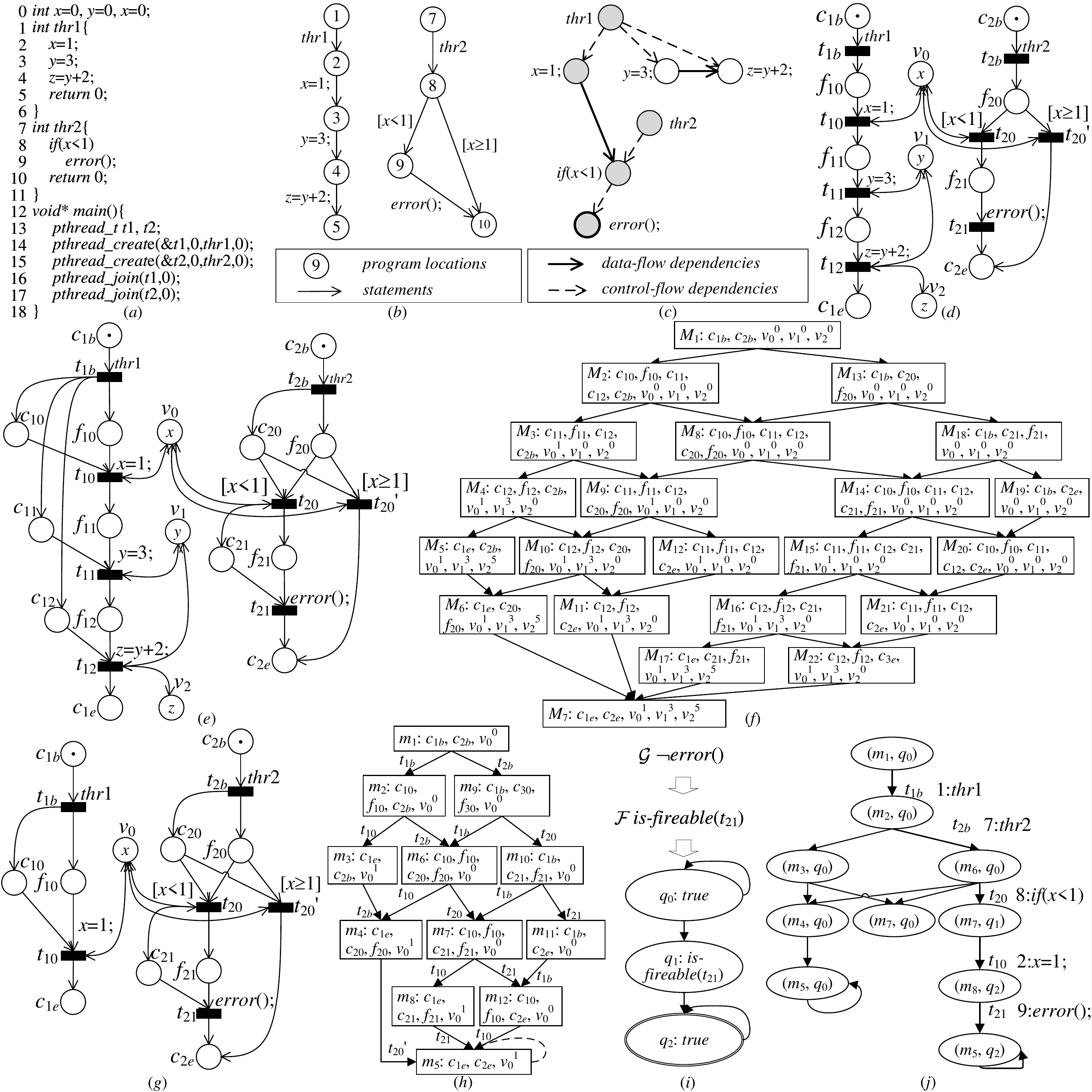}
 \caption{Motivating example ($a$) A concurrent program with an error location ($b$) The CFA of this program ($c$) The PDG of this program ($d$) A CPN converted from this program ($e$) A PDNet converted from this program ($f$) The state-space of the PDNet ($g$) The PDNet slice for $\mathcal{G}$ $\neg error()$ ($h$) The state space of the PDNet slice ($i$) The B\"{u}chi automaton for the negation of $\mathcal{G}$ $\neg error()$  ($j$) The explored product automaton of PDNet slice for $\mathcal{G}$ $\neg error()$}\label{Fg:Example}
\end{figure*}

To explore the challenge of verifying LTL in concurrent programs that utilize POSIX threads \cite{PThread2019}, we present an example program in Figure \ref{Fg:Example}($a$), which contains an error on Line $9$. The safety of the program is defined by the LTL-$_\mathcal{X}$ formula $\mathcal{G}$ $\neg error()$, which ensures that the function $error()$ is not executed in any state along any path.
CFA in Figure \ref{Fg:Example}($b$) shows the program's control-flow structure. The nodes are represented as circles with integers (matching those in Figure \ref{Fg:Example}($a$)), and the edges are represented by arrows indicating the statements.
For instance, the node labeled by $9$ corresponds to Line $9$ of the program, and the edge from $9$ to $10$ corresponds to $error()$.
PDG \cite{Hatcliff1999A} is in Figure \ref{Fg:Example}($c$).
For control-flow dependencies represented by dotted arrows, $x$$=$$1$, $y$$=$$3$ and $z$$=$$y$$+$$2$ all depend on the entry of $thr1$ $error()$ depends on the condition $if(x$$<$$1)$, while $if(x$$<$$1)$ depends on the entry of $thr2$.
For data-flow dependencies represented by bold arrows, $if(x$$<$$1)$ of $thr2$ depends on $x$$=$$1$ of $thr1$ because $x$$<$$1$ references $x$ defined in $x$$=$$1$ and they belong to different concurrently executing threads. $z$$=$$y$$+$$2$ depends on $y$$=$$3$ of $thr1$ because $z$$=$$y$$+$$2$ references $y$ defined in $y$$=$$3$ and $z$$=$$y$$+$$2$ is reachable from $y$$=$$3$.
The criterion for $\mathcal{G}$ $\neg error()$ is represented by bold nodes, and the remaining nodes are filled with light gray in Figure \ref{Fg:Example}($c$). 

Figure \ref{Fg:Example}($d$) gives the traditional CPN model \cite{Jensen2009Formal} for this program. For the sake of simplicity, the labels on the arcs are not shown. In this model, each transition signifies the execution of a statement when it occurs. The variables represented by the places are operated upon by the corresponding transition occurrences. For instance, after the firing of $t_{10}$ ($x$$=$$1$), $t_{11}$ ($y$$=$$3$) is enabled. After $t_{10}$ occurs, $x$ is assigned to $1$. However, it is important to note that this model only represents the control-flow structure and does not show program dependencies.
To combine the control-flow structure and control-flow dependencies, we establish specific places and arcs within the PDNet transition illustrated in Figure \ref{Fg:Example}($e$). Although the labels on the arcs have been excluded for clarity purposes, they are essential in distinguishing between a statement's execution order condition and domination condition.
The former represents the actual execution syntax and semantics in the control-flow structure, while the latter represents the control-flow dependencies.
For instance, $f_{11}$, $(t_{10},f_{11})$ and $(f_{11},t_{11})$ characterize the fact that $y$$=$$3$ executes after $x$$=$$1$.
$c_{11}$, $(t_{1b}, c_{11})$ and $(c_{11},t_{11})$ characterize the fact that $y$$=$$3$ depends on the entry of $thr1$.

The state space depicted in Figure \ref{Fg:Example}($f$) is represented by the reachability graph of the PDNet. The graph's nodes are labeled rectangles that correspond to place names, signifying markings. The edges are arrows that denote the fired transitions of the PDNet, with transition labels are excluded for clarity. The presence of tokens in places represents the markings, such as the marking $M_7$ where tokens are in the places $c_{1e}$, $c_{2e}$, $v_0$, $v_1$, and $v_2$.
Due to the firing of $t_{20}'$, $(M_6,M_7)$ takes place. The superscripts assigned to places $v_0$, $v_1$, and $v_2$ indicate the corresponding variable values in this marking. For example, the symbol $v_0^1$ in $M_7$ represents that the value of $x$ is $1$.

Moreover, superfluous computation cost may arise from data-flow dependencies pertaining to irrelevant variables. As illustrated in Figure \ref{Fg:Example}($c$), the PDG captures data-flow dependencies associated with $y$, which are not encompassed in any slice. Conversely, data-flow dependencies can only be noted when the relevant variable is verified to be unsliced. To avoid such expense, we propose PDNet slicing for capturing data-flow dependencies solely when required.
In Figure \ref{Fg:Example}($g$), the PDNet slice effectively captures the data-flow dependencies of $x$, leading to a reduction in several variables such as $c_{11}$, $f_{11}$, $t_{11}$, $c_{12}$, $f_{12}$, $t_{12}$, $v_{1}$, and $v_{2}$ as compared to traditional PN slicing methods. Consequently, $y=3$ and $z=y+2$ are sliced away. The state space of the PDNet slice, which is represented by the marking graph, is depicted in Figure \ref{Fg:Example}($h$). Transitions labeled on the edge signify the fired transition. For example, $m_7\stackrel{t_{10}}{\longrightarrow}m_8$ means $t_{10}$ fired under $m_7$ and $m_8$ is generated. As compared to the state space shown in Figure \ref{Fg:Example}($f$), PDNet slicing reduces $10$ states from ones which is significant. Additionally, a dotted arrow is added to an arc of $M_5$ pointing to itself because the LTL-$_\mathcal{X}$ model checking is based on the infinite path \cite{Kahlon2006An}.

In order to establish the safety property of the example program through the LTL-$_\mathcal{X}$ formula $\mathcal{G}$ $\neg error()$, we begin by converting $error()$ into $is\mbox{-} fireable(t_{21})$ of the PDNet. When $t_{21}$ is enabled in a marking, this proposition is considered to be $true$. For instance, in $m_{10}$ of Figure \ref{Fg:Example}($h$), $t_{21}$ is enabled and $is\mbox{-} fireable(t_{21})$ is $true$ in $m_{10}$. Next, we convert the formula to its negation form $\mathcal{F}$ $is\mbox{-} fireable(t_{21})$, which is then translated to a B\"{u}chi automaton \cite{Ding2023Enpac} depicted in Figure \ref{Fg:Example}($i$). The B\"{u}chi automaton features three states, where $true$ pertains to nodes $q_0$ and $q_2$ that can be synchronized with any reachable markings.
Given Figure \ref{Fg:Example}($h$), only the feasible markings that enable $t_{21}$ can synchronize with $q_1$, identified as $is\mbox{-} fireable(t_{21})$. In the on-the-fly exploration \cite{Ding2023Enpac}, the first counterexample in Figure \ref{Fg:Example}($j$) assesses $10$ product states only. One of these states, $(m_7,q_1)$, synchronizes marking $m_7$ in Figure \ref{Fg:Example}($h$) and state $q_1$ in Figure \ref{Fg:Example}($i$) because $t_{21}$ is enabled in $m_7$. Consequently, the example program violates the safety property $\mathcal{G}$ $\neg error()$ in Figure \ref{Fg:Example}($a$). The execution of statements $1$, $7$, $8$, $2$ and $9$, corresponding to the occurrence sequence $t_{1b}$, $t_{2b}$, $t_{20}$, $t_{10}$ and $t_{21}$ identified by the marking sequence $m_1,m_2,m_6,m_7,m_8,m_5,\cdots$ in Figure \ref{Fg:Example}($j$), serves as a counterexample.

\section{Program Dependence Net (PDNet)}\label{Sec:PDNet}
\subsection{PDNet}\label{Sub:PDNet}
In the following, $\mathbb{B}$ is the set of Boolean predicates with standard logic operations, $\mathbb{E}$ is a set of expressions, $Type[e]$ is the type of an expression $e$$\in$$\mathbb{E}$, i.e., the type of the values obtained when evaluating $e$, $Var(e)$ is the set of all variables in an expression $e$, $\mathbb{E}_V$ for a variable set $V$ is the set of expressions $e$$\in$$\mathbb{E}$ such that $Var(e)$$\subseteq$$V$, $Type[v]$ is the type of a variable $v$$\in$$V$, $\mathbb{O}$ is the set of constants, and $Type[o]$ is the type of constant $o$$\in$$\mathbb{O}$.

\begin{definition}[PDNet]\label{Def:PDNet}
PDNet is defined as a $9$-tuple $N$ $::=$ $(\Sigma$, $V$, $P$, $T$, $F$, $C$, $G$, $E$, $I)$, where:
	
1. $\Sigma$ is a finite non-empty set of types called color sets.
	
2. $V$ is a finite set of typed variables. $\forall v$$\in$$V\colon Type[v]$$\in$$\Sigma$.
	
3. $P=P_c\cup P_v\cup P_f$ is a finite set of places.  $P_c$ is a subset of control places, $P_v$ is a subset of variable places, and $P_f$ is a subset of execution places.
	
4. $T$ is a finite set of transitions and $T\cap P=\emptyset$.
	
5. $F\subseteq(P\times T)\cup(T\times P)$ is a finite set of directed arcs. $F=F_c\cup F_{rw}\cup F_f$. Concretely, $F_c\subseteq(P_c\times T)\cup(T\times P_c)$ is a subset of control arcs, $F_{rw}\subseteq(P_v\times T)\cup(T\times P_v)$ is a subset of read-write arcs, and $F_f\subseteq(P_f\times T)\cup(T\times P_f)$ is a subset of execution arcs.
	
6. $C\colon P$$\rightarrow$$\Sigma$ is a color set function that assigns a color set $C(p)$ belonging to the set of types $\Sigma$ to each place $p$.
	
7. $G\colon T$$\rightarrow$$\mathbb{E}_V$ is a guard function that assigns an expression $G(t)$ to each transition $t$. $\forall t$$\in$$T\colon Type[G(t)]\in BOOL)\wedge(Type[Var(G(t))]\subseteq\Sigma$.
	
8. $E\colon F$$\rightarrow$$\mathbb{E}_V$ is a function that assigns an arc expression $E(f)$ to each arc $f$. $\forall f$$\in$$F$$:$ $(Type[E(f)]=C(p(f))_{MS})\wedge(Type[Var(E(f))]\subseteq\Sigma)$, where $p(f)$ is the place connected to arc $f$.
	
9. $I\colon P$$\rightarrow$$\mathbb{E}_\emptyset$ is an initialization function that assigns an initialization expression $I(p)$ to each place $p$. $\forall p$$\in$$P\colon Type[I(p)]$$=$$C(p)_{MS})\wedge(Var(I(p))$$=$$\emptyset$.
\end{definition}

PDNet differs from CPNs in $P$ and $F$. Control places $P_c$ with their adjacent control arcs $F_c$ are used to model dominant relationships of control-flow dependencies, variable places $P_v$ with their adjacent read-write arcs $F_{rw}$ are used to model variables and their read/write relationships, and execution places $P_f$ with their adjacent execution arcs $F_f$ are used to model execution relationships in control-flow structures in concurrent programs. Other definitions and constraints of PDNet are consistent with CPN's.

As the example in Figure \ref{Fg:Example}($g$), $P_v$$=$$\{v_0\}$ where $v_0$ is a variable place corresponding to variable $x$, $P_c$$=$$\{c_{1b}$$, $
$c_{2b}$$,$ $c_{1e}$$, $ $c_{2e}$$, $ $c_{10}$$, $ $c_{20}$$, $ $c_{21}\}$, $P_f$$ = $$\{c_{1b}$$, $ $c_{2b}$$, $ $c_{1e}$$, $ $c_{2e}$$, $ $f_{10}$$, $ $f_{20}$$, $ $f_{21}\}$, and $t_{21}$ corresponds to the statement $error()$.
For a node $x \in P\cup T$, its preset $^\bullet$$x=\{y \vert (y, x) $$\in$$ F\}$ and its postset $x$$^\bullet = \{y \vert (x, y) $$\in$$ F\}$ are two subsets of $P\cup T$. For instance, $^\bullet$$t_{21} = \{f_{21}$$, $ $c_{21}\}$, $t_{21}$$^\bullet = \{c_{2e}\}$, and $^\bullet$$v_0 = v_0$$^\bullet$$=$$\{t_{10}$$,$ $t_{20} $$,$ $t_{20}'\}$ in Figure \ref{Fg:Example}($e$).
Some basic concepts of PDNet are defined next.

\begin{definition}\label{Def:Con}
Let $N$ be a PDNet.

1. $M\colon P \rightarrow \mathbb{E}_\emptyset$ is a marking function that assigns an expression $M(p)$ to each place $p$. $\forall p \in P\colon Type[M(p)]=C(p)_{MS}\wedge(Var(M(p))=\emptyset)$.
$M_0$ represents the initial marking, i.e., $\forall p \in P\colon M_0(p)=I(p)$.

2. $Var(t)\subseteq V$ is the variable set of transition $t$. It consists of the variables appearing in expression $G(t)$ and in arc expressions of all arcs connected to $t$.

3. $B\colon V \rightarrow \mathbb{O}$ is a binding function that assigns a constant value $B(v)$ to variable $v$. $B[t]$ presents the set of all bindings for transition $t$, that maps $v$$\in$$Var(t)$ to a constant value, and $b$$\in$$B[t]$ is a binding of $t$.

4. A binding element $(t,b)$ is a pair where $t$$\in$$T$ and $b$$\in$$B[t]$. $\widetilde{\mathbb{B}}(t)$ is a set of all binding elements of $t$.
\end{definition}

For convenience, a marking of $N$ is denoted by $M$ or $M$ with a subscript.
Then, guard and arc expressions are evaluated as follows.
Formally, $e\langle b\rangle$ represents the evaluation result of expression $e$ in binding $b$ by assigning a constant from $b$ to variable $v$$\in$$Var(e)$.
Thus, under a binding element $(t,b)$$\in$$\widetilde{\mathbb{B}}(t)$, $G(t)\langle b\rangle$ (or $E(f)\langle b\rangle$) represents the evaluation result of $G(t)$ (or $E(f)$), where $f$ is an arc connected to $t$.
For instance, $[x$$<$$1]$$\langle\{b(x)$$=$$1\}\rangle$$=$$flase$ for guard expression $G(t_{20})$$=$$[x$$<$$1]$ in Figure \ref{Fg:Example}($e$).

\begin{definition}[Enabling and Occurrence Rules of PDNet]\label{Def:Ena}
Let $N$ be a PDNet, $(t,b)$ a binding element, and $M$ a marking.
A binding element $(t,b)$ is enabled under $M$, denoted by $M[(t,b)\rangle$, if

1. $G(t)\langle b\rangle=true$, and

2. $\forall p\in^\bullet$$t\colon E(p,t)\langle b\rangle\leq M(p)$.

When $(t,b)$ is enabled under $M$, it can occur and lead to a new marking $M_1$ of $N$, denoted by $M[(t,b)\rangle M_1$, such that
$\forall$$p\in$$P\colon M_1(p)=M(p)-E(p,t)\langle b\rangle +E(t,p)\langle b\rangle$.
\end{definition}

For ease of expression, if binding $b$$\in$$B[t]$ enables binding element $(t,b)$$\in$$\widetilde{\mathbb{B}}(t)$ under a marking, we call $t$ enabled and can occur or this marking can fire $t$.
For instance, under marking $m_7$ in Figure \ref{Fg:Example}($f$), $t_{21}$ is enabled because $G(t_{21})$$=$$true$, $f_{21}$ and $c_{21}$ belonging to $^\bullet$$t_{21}$ are marked in $m_7$, and $E(f_{21},t_{21})$$\leq$$M(f_{21})$ and $E(c_{21},t_{21})$$\leq$$M(c_{21})$ are satisfied.
Given $b_{21}$$\in$$B[t_{21}]$, $m_7[(t_{21},b_{21})\rangle m_{12}$ where $c_{2e}$$\in$$t_{21}$$^\bullet$ is marked in $m_{12}$. That is, $m_7$ can fire $t_{21}$.

\begin{definition}[Occurrence Sequence of PDNet]\label{Def:Seq}
Let $N$ be a PDNet, $M_0$ be the initial marking, and $(t,b)$ be a binding element. An occurrence sequence $\omega$ of $N$ can be defined by the following inductive scheme:
1) $M_0[\varepsilon\rangle M_0$$(\varepsilon$ is an empty sequence$)$, and
2) $M_0[\omega$$\rangle$$M_1$$\wedge$$M_1[(t,b)\rangle M_2$$\colon$$M_0[$$\omega$$(t,b)\rangle M_2$.
An occurrence sequence $\omega$ of $N$ is maximal, if
1) $\omega$ is of infinite length (e.g., $(t_1,b_1)$, $(t_2,b_2)$, $\cdots$, $(t_n,b_n)$, $\cdots$), or
2) $M_0[\omega\rangle M_1$$\wedge$$\forall t$$\in$$T$, $\nexists(t,b)$$\in$$BE(t)$$:M_1[(t,b)\rangle$.
\end{definition}


A marking sequence w.r.t. $\omega$, denoted by $M$[$\omega$] or $M_1$, $M_2$, $\cdots$, and $M_n$, is generated by occurrence of all binding elements in $\omega$. For ease of expression, $\omega$ can be represented by $\langle t_{1},t_{2},\cdots,t_{n}\rangle$.
For instance, $\langle t_{1b},t_{2b},t_{20},t_{10},t_{21}\rangle$ is a maximal occurrence sequence of the PDNet in Figure \ref{Fg:Example}($e$).
And $\langle m_1,m_2,m_6,m_7,m_8,m_5\rangle$ is a marking sequence by firing $t_{1b}$, $t_{2b}$, $t_{20}$, $t_{10}$ and $t_{21}$ one after another.

\subsection{Linear Temporal Logic of PDNet}\label{Sub:LTL}
The LTL formalism, elucidated in \cite{Ding2023Enpac}, serves as a specification for delineating linear temporal properties, encompassing safety and liveness properties of Petri nets. Nevertheless, if slicing methods are employed to condense the state space, the resultant model may not encompass the entire sequence of the original model. Consequently, our methods are not aligned with operator $\mathcal{X}$ and can corroborate LTL-$_\mathcal{X}$ formulae.

\begin{definition}[Proposition and LTL-$_\mathcal{X}$ Formula of PDNet]\label{Def:LTL}
Let $N$ be a PDNet, $a$ be a proposition, $\mathbb{A}$ be a set of propositions, and $\psi$ be an LTL-$_\mathcal{X}$ formula. The syntax of propositions is defined as:

\begin{equation}
a ::= true \vert false \vert is\mbox{-} fireable(t) \vert token\mbox{-} value(p)\ r\ c
\end{equation}

where $t\in T, p\in P_v, c\in C(p)_{MS}$ is a constant, $r$$\in$$\{<,\leq,>,\geq,=\}$.

The proposition semantics is defined w.r.t a marking $M$:

\begin{equation}\centering
    	is\mbox{-} fireable(t) = \left\{
		\begin{aligned}
			&true    &if\ \exists b\colon M[(t,b)\rangle, \\
			&false   &otherwise.
		\end{aligned}
	\right.
\end{equation}

\begin{equation}
	token\mbox{-} value(p)\ r\ c = \left\{
		\begin{aligned}
			&true    &if\ M(p)\ r\ c, \\
			&false   &otherwise.
		\end{aligned}
	\right.
\end{equation}



The syntax of LTL-$_\mathcal{X}$ over $\mathbb{A}$ is defined as:

\begin{equation}
\psi ::= a \vert \neg\psi \vert \psi_1\wedge\psi_2 \vert \psi_1\vee\psi_2 \vert \psi_1\Rightarrow\psi_2 \vert \mathcal{F}\psi \vert \mathcal{G}\psi \vert  \psi_1\mathcal{U}\psi_2
\end{equation}


where $\neg$, $\wedge$, $\vee$ and $\Rightarrow$ are usual propositional connectives, $\mathcal{F}$, $\mathcal{G}$ and $\mathcal{U}$ are temporal operators, $\psi$, $\psi_1$ and $\psi_2$ are LTL-$_\mathcal{X}$ formulae.
\end{definition}

The condensed explanation of LTL-$_\mathcal{X}$ semantics under a marking sequence is similar to that of Petri nets as explained in \cite{Wolf2019Petri}. For example, if $\mathcal{G}$ $is\mbox{-} fireable(t)\Rightarrow \mathcal{F}$ $token\mbox{-} value(p)$$=$$0$, it means that whenever $t$ is enabled, the token of $p$ is zero in some subsequent states. Additionally, the B\"{u}chi automaton can be used to encode an LTL-$_\mathcal{X}$ formula for explicit-state model checking, as described in \cite{Ding2023Enpac} and illustrated in Figure \ref{Fg:Example}($i$) \cite{Kahlon2006An}.

For explicit-state model checking, the traditional approach has been automaton-theoretic. This involves exhaustively exploring all possible transition firings of a transition system (state space). For LTL, the problem is translated into an emptiness-checking problem. PDNet's analysis can also adopt the automaton-theoretic approach. Specifically, the marking of PDNet can synchronize with the states of B\"{u}chi automaton \cite{Ding2023Enpac}, which are known as B\"{u}chi states. To start, the initial product state is generated by the initial marking and the initial B\"{u}chi state. Then, an acceptable path starting from the initial product state is extended until reaching an acceptable product state \cite{Ding2023Enpac}. Thus, $N$$\models$$\psi$ holds if no acceptable sequence is reachable from the initial product state. For instance, there exists an explored counterexample in Figure \ref{Fg:Example}($j$), meaning that $N$$\models$$\psi$ does not hold.

\section{Dependency Modeling Based on PDNet}\label{Sec:Model}
\subsection{PDNet Modeling for Concurrent Program}\label{Sub:Sem}
\begin{table}[h]\centering
\caption{PDNet Transition}\label{Tab:PDNetTransition}
\begin{tabular}{ccc}
\toprule[1pt]
Action & Transition  & Operation \cite{Ding2023PDNet}   \\ \midrule
$ass$  & $assign$ transition                  & $\nu:=w$ \\ \midrule
$jum$  & $jump$ transition                    & \multirow{2}{*}{$jump$}   \\
$ret$  & $exit$ transition                    &    \\ \midrule
$tcd$  & \multirow{2}{*}{$branch$ transition} & $if(w)$$then$$(tau_1*)$$else$$(tau_2*)$ \\
$fcd$  &                                      & $while(w)$$do$$(tau*)$   \\ \midrule
$call$ & $call$ transition                    & \multirow{2}{*}{$calls$}   \\
$rets$ & $return$ transition                  &  \\ \midrule
$acq$  & $lock$ transition                    & $\langle lock, \ell\rangle$  \\ \midrule
$rel$  & $unlock$ transition                  & $\langle unlock, \ell\rangle$ \\ \midrule
$sig$  & $signal$ transition                  & $\langle signal, \gamma \rangle$   \\ \midrule
$wa_1$ & \multirow{3}{*}{$wait$ Transition}   & \multirow{3}{*}{$\langle wait, \gamma, \ell \rangle$}   \\
$wa_2$ &  \\
$wa_3$ &  & \\
\bottomrule[1pt]
\end{tabular}
\end{table}

We explain the syntax and semantics of concurrent programs with function calls \cite{Masud2021Semantic} and concurrency primitives \cite{PThread2019}, as outlined in \cite{Ding2023PDNet}. There are $13$ different actions, including $ass$, $jum$, $ret$, $tcd$, $fcd$, $call$, $rets$, $acq$, $rel$, $sig$, $wa_1$, $wa_2$ and $wa_3$ \cite{Ding2023PDNet}. Each of these actions corresponds to a specific transition as shown in Table \ref{Tab:PDNetTransition}, when modeling the control-flow structure of concurrent programs using PDNet.

\begin{figure*}[!ht]\centering
\includegraphics[width=\textwidth]{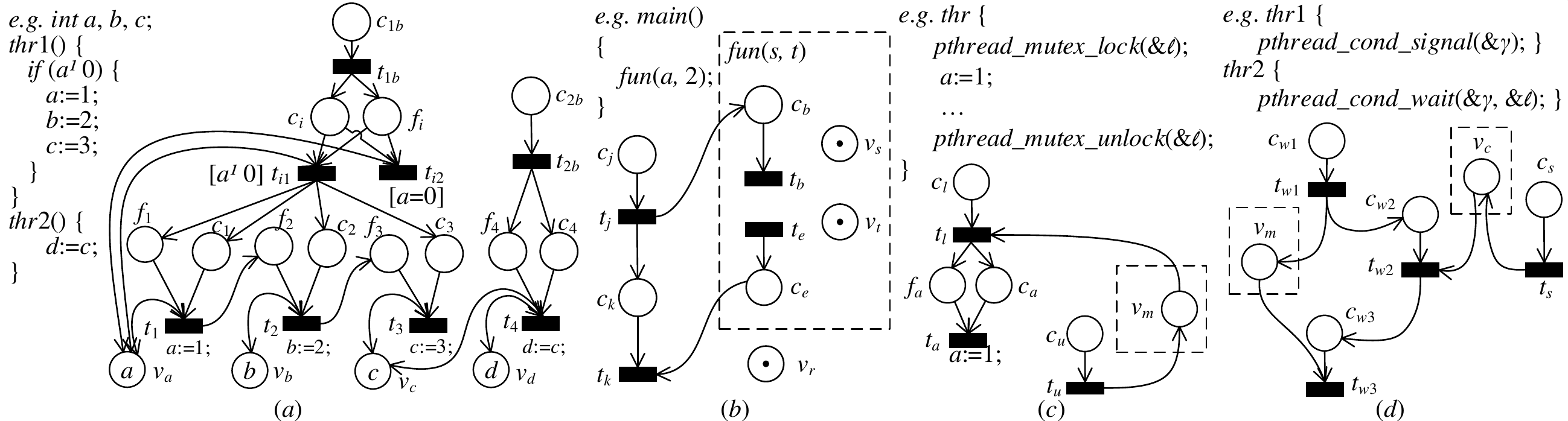}
  \caption{Example for Control-flow Structure ($a$) Example for $ass$, $jum$, $ret$, $tcd$ and $fcd$ ($b$) Example for $call$ and $rets$ ($c$) Example for $acq$ and $rel$ ($d$) Example for $sig$, $wa_1$, $wa_2$ and $wa_3$}\label{Fg:Control}
\end{figure*}

For instance, $t_1$ for $a$$:=$$1$, $t_2$ for $b$$:=$$2$, $t_3$ for $c$$:=$$3$ and $t_4$ for $d$$:=$$c$ in Figure \ref{Fg:Control}($a$) are $assign$ transitions, and $t_{i1}$ and $t_{i2}$ are $branch$ transitions for $if(a$$\neq$$0)$.
In Figure \ref{Fg:Control}($b$), $t_j$ is $call$ transition, and $t_k$ is $return$ transition.
In Figure \ref{Fg:Control}($c$), $t_l$ is $lock$ transition for \emph{pthread\_mutex\_lock}(\verb'&'$\ell$), and $t_u$ is $unlock$ transition for \emph{pthread\_mutex\_unlock}(\verb'&'$\ell$).
In Figure \ref{Fg:Control}($d$), $t_s$ is $signal$ transition for \emph{pthread\_cond\_signal}(\verb'&'$\gamma$), $t_{w1}$ for action $wa_1$, $t_{w2}$ for action $wa_2$ and $t_{w3}$ for action are $wait$ transitions from \emph{pthread\_cond\_wait}(\verb'&'$\gamma$, \verb'&'$\ell$).

Specially, $t_b$ is called $enter$ transition of the function $fun(s,t)$, $t_e$ is called $exit$ transition of $fun(s,t)$, ($t_j,c_b$) is called $enter$ arc, and ($c_e,t_k$) is called $exit$ arc in Figure \ref{Fg:Control}($b$).

\subsection{Control-flow Dependencies Based on PDNet}\label{Sub:ConDep}

Control-flow dependencies among statements were characterized as domination, as noted in \cite{Masud2021Semantic}. However, the definition of control-flow dependencies differs among such graphs as those in \cite{Ferrante1987The,Horwitz1990Interprocedural,Krinke2004Advanced,Qi2017Precise}. In this paper, we classify control-flow dependencies into four types: control, call, lock, and prior-occurrence ones. For a complete and cohesive representation of them, we utilize PDNet as a comprehensive way.

As mentioned above, $branch$ transition is modeled for actions $tcd$ or $fcd$, $lock$ ($unlock$) transition is modeled for $acq$ ($rel$), and $signal$ ($wait$) transition is modeled for $sig$ ($wa_1,wa_2,wa_3$) in Table \ref{Tab:PDNetTransition}. $enter$ ($exit$) transition is constructed for a function like $t_b$ ($t_e$) in Figure \ref{Fg:Control}($b$).
Next, we define the execution path, control scope and critical region of PDNet.

\begin{definition}[Execution Path of PDNet]\label{Def:EP}
Let $N$ be a PDNet, and $t_m$ and $t_n$ be the two different transitions of $N$. A sequence $\pi$ along the transitions and places of $N$ is denoted by ($t_1$, $p_1$, $t_2$, $p_2$, $\dots$, $t_{k-1}$, $p_{k-1}$, $t_k$), where $(t_i,p_{i})$ or $(p_i,t_{i+1})$ ($1$$\leq$$i$$<$$k$) is an arc of $N$. $\pi$ is an execution path from $t_m$ to $t_n$ if
$t_1$ is $t_m$, $t_k$ is $t_n$, and $k$$\geq$$1$. 

All execution paths from $t_m$ to $t_n$ constitute the execution path set, denoted by $\mathbb{P}(t_m,t_n)$.
The transition $t_n$ is reachable from $t_m$, denoted by $\mathbb{R}(t_m,t_n)$, if
$\mathbb{P}(t_m,t_n)$$\neq$$\emptyset$.
Particularly, $t_n$ is reachable from $t_m$ irrelevant to arc $(t_i,p_i)$, denoted by $\mathbb{R}(t_m,t_n)_{\vec{t_i p_i}}$, if $\forall\pi$$\in$$\mathbb{P}(t_m,t_n)$, $(t_i,p_i)$$\notin$$\pi$.
\end{definition}

\begin{definition}[Control Scope of PDNet]\label{Def:CS}
Let $N$ be a PDNet, $t_m$ be a $branch$ or $enter$ transition of $N$, and $t_n$ be a transition of $N$.
$t_n$ is in the control scope of $t_m$, denoted by $\tilde{\mathbb{C}}(t_m,t_n)$, if
$t_n$ is reachable from $t_m$, and there exists an execution path starting in $t_m$ such that it does not contain $t_n$.
\end{definition}

\begin{definition}[Critical Region of PDNet]\label{Def:CR}
Let $N$ be a PDNet, $t_m$ be a $lock$ transition of $N$ that can acquire a mutex $\ell$, and $t_n$ be a transition of $N$.
$t_n$ is in the critical region of $t_m$, denoted by $\check{\mathbb{C}}(t_m,t_n)$, if
$t_n$ is reachable from $t_m$, and there exists no $unlock$ transition that releases the same mutex $\ell$ in any execution path in $\mathbb{P}(t_m,t_n)$.
\end{definition}

In Figure \ref{Fg:Control}($a$), $\langle t_{i1}$, $f_1$, $t_1$, $f_2$, $t_2$, $f_3$, $t_3\rangle$) is an execution path. $t_3$ is reachable from $t_{i1}$, i.e., $\mathbb{R}(t_{i1},t_3)$. $t_1$, $t_2$ and $t_3$ are in the control scope of $t_{i1}$, i.e., $\tilde{\mathbb{C}}(t_{i1},t_1)$, $\tilde{\mathbb{C}}(t_{i1},t_2)$ and $\tilde{\mathbb{C}}(t_{i1},t_3)$.
In Figure \ref{Fg:Control}($c$), $t_a$ is in the critical region of $t_l$, i.e., $\check{\mathbb{C}}(t_l,t_a)$.
Next, we define four kinds of control-flow dependencies.

\begin{definition}[Control-flow Dependencies of PDNet]\label{Def:CFD}
For concurrent program $\mathcal{P}$, $N$ is the PDNet of $\mathcal{P}$, $t_m$ and $t_n$ are two transitions of $N$:

1. $t_m$ is control-dependent on $t_n$, denoted by $t_n$$\stackrel{co}{\longrightarrow}$$t_m$,
if
1) $t_n$ is a $branch$ transition or $enter$ transition,
2) $\tilde{\mathbb{C}}(t_n,t_m)$, and
3) there exists no other $branch$ transition $t_o$ in the control scope of $t_n$ such that $\tilde{\mathbb{C}}(t_o,t_m)$.

2. $t_m$ is call-dependent on $t_n$, denoted by $t_n$$\stackrel{ca}{\longrightarrow}$$t_m$,
if $t_m$ is an $enter$ transition of the called function, and
$t_n$ is the $call$ transition of the calling function, making the execution flow turn to $t_m$, or
$t_n$ is an $exit$ transition of a called function, and $t_m$ is the $return$ transition of a calling function, making the execution flow turn back to $t_m$.

3. $t_m$ is lock-dependent on $t_n$, denoted by $t_n$$\stackrel{lo}{\longrightarrow}$$t_m$,
if $t_n$ is a $lock$ transition acquiring $\ell$, and $\check{\mathbb{C}}(t_n,t_m)$, or $t_m$ and $t_n$ are all $lock$ transitions that acquire $\ell$. 

4. $t_m$ is prior-occurrence-dependent on $t_n$, denoted by $t_n$$\stackrel{po}{\longrightarrow}$$t_m$,
if
1) $t_m$ is a $wait$ transition waiting for a condition variable $\gamma$, and
2) $t_n$ is a $signal$ transition notifying on the same condition variable $\gamma$.
\end{definition}

Intuitively, we define the control dependence for the nearest $branch$ or $enter$ transition of PDNet. A $branch$ or $enter$ transition can dominate the execution of the following transitions that are not in the control scope of other transitions.
As mentioned above, the PDNet control-flow structure provides the control place interfaces (e.g., $c_{1}$, $c_{2}$ and $c_{3}$ in Figure \ref{Fg:Control}($a$)) to describe the control dependencies.
The control arcs between a $branch$ or $enter$ transition and control places are constructed for the control dependencies.

In Figure \ref{Fg:Control}($a$), $a$, $b$ and $c$ are global variables initialized to $1$ in this program.
The PDNet structures of the branching operation $if(a$$\neq$$0)$ and assignment operations $a$$:=$$1$, $b$$:=$$2$ and $c$$:=$$3$ are constructed by modeling these operations. 
Thus, the control arcs ($t_{i1}$, $c_1$), ($t_{i1}$, $c_2$) and ($t_{i1}$, $c_3$) are constructed to describe the control dependencies based on the Definition \ref{Def:CFD}.
That is, the control dependencies $t_{i1}$$\stackrel{co}{\longrightarrow}$$t_1$, $t_{i1}$$\stackrel{co}{\longrightarrow}$$t_2$ and $t_{i1}$$\stackrel{co}{\longrightarrow}$$t_3$ are represented explicitly.

Other control-flow dependencies have been described by modeling function call operations and POSIX thread operations.
For example, the call dependencies $t_j$$\stackrel{ca}{\longrightarrow}$$t_b$ and $t_e$$\stackrel{ca}{\longrightarrow}$$t_k$ are represented in Figure \ref{Fg:Control}($b$), the lock dependence $t_l$$\stackrel{lo}{\longrightarrow}$$t_a$ is represented in Figure \ref{Fg:Control}($c$), and the prior-occurrence dependence $t_{w2}$$\stackrel{po}{\longrightarrow}$$t_s$ is represented in Figure \ref{Fg:Control}($d$).

\subsection{Data-flow Dependencies Based on PDNet}\label{sub:DataDep}

Data-flow dependencies describe the reachable definition-use relation of the variables. We classify the data-flow dependencies into data \cite{Ferrante1987The} and interference ones \cite{Qi2017Precise}.
Next, we define a reference and definition set of PDNets.

\begin{definition}[Reference Set and Definition Set of PDNet]\label{Def:defref}
Let $N$ be a PDNet and $t$ be a transition of $N$. The reference set of $t$ $Ref(t) ::= \{p \vert\forall p \in ^\bullet$$t\cap P_v\colon E(p,t) = E(t,p)\}$.
The definition set of $t$ $Def(t) ::=
\{p \vert \forall p \in ^\bullet$$t\cap P_v\colon E(p,t) \neq E(t,p)\}$.
\end{definition}

\begin{definition}[Data-flow Dependencies of PDNet]\label{Def:DFD}
For concurrent program $\mathcal{P}$, $N$ is the PDNet of $\mathcal{P}$, $t_m$ and $t_n$ are two transitions of $N$, if there is a variable place $v$, such that

1. $t_m$ is data-dependent on $t_n$, denoted by $t_n$$\stackrel{D}{\longrightarrow}$$t_m$,
if
1) $\mathbb{R}(t_n,t_m)$, and
2) $v$$\in$$Ref(t_m)\wedge v$$\in$$Def(t_n)$,
3). $\exists \pi\in \mathbb{P}(t_n,t_m)$, there exists no other transition $t_a$$\in$$\pi$ such that $v$$\in$$Def(t_a)$.

2. $t_m$ is interference-dependent on $t_n$, denoted by $t_n$$\stackrel{I}{\longrightarrow}$$t_m$,
if
1) there exist execution places $f_m$$\in$$(^\bullet t_m$$\cap$$P_f)$ and $f_n$$\in$$(^\bullet t_n$$\cap$$P_f)$ such that $M(f_m)$$\neq$$M(f_n)$, and
2) $v$$\in$$Ref(t_m)\wedge v$$\in$$Def(t_n)$.
\end{definition}

In Definition \ref{Def:DFD}'s $2.1)$, $M(f_m)$$\neq$$M(f_n)$ means that the operations corresponding to $t_m$ and $t_n$ belong to different concurrently executing threads.
For instance, $t_{11}$$\stackrel{D}{\longrightarrow}$$t_{12}$ is represented in Figure \ref{Fg:Example}($e$).
$t_{10}$$\stackrel{I}{\longrightarrow}$$t_{20}$ is represented in Figure \ref{Fg:Example}($e$), and $t_{3}$$\stackrel{I}{\longrightarrow}$$t_4$ in Figure \ref{Fg:Control}($a$).
Through the use of read-write arcs and control-flow structures in PDNet, we can capture these data-flow dependencies when needed without adding additional arcs. When slicing, the read-write arcs in PDNet can be used to capture these dependencies on demand. It is worth noting that local variables exhibit only data dependencies between $t_m$ and $t_n$, whereas $t_m$ may be data or interference dependent on $t_n$ if $v$ represents a global variable. This approach can help reduce computation cost.

\section{On-demand PDNet Slicing for Reduction}\label{Sec:Slice}
\subsection{Slicing Criterion of PDNet}\label{Sub:Crit}
The concept behind PDNet slicing is to eliminate unnecessary parts that are not relevant to the verified property. This reduces both model size and reachable state space. The slicing criterion is determined by the program features specified in the verified LTL-$_\mathcal{X}$ property. As a result, we extract the relevant slicing criterion for an LTL-$_\mathcal{X}$ formula.

\begin{definition}[Slicing Criterion]\label{Def:Crit}
Let $N$ be a PDNet, $\psi$ a LTL-$_\mathcal{X}$ formula with propositions in Definition \ref{Def:LTL}, $\mathbb{A}$ the proposition set from $\psi$, and $Crit$ the slicing criterion w.r.t. $\psi$.
If $a$ is in the form of $is\mbox{-} fireable(t)$, $Crit(a)$ $::=$ \{$p\vert p$$\in$$^\bullet$$t$$\setminus$$P_f$\}.
If $a$ is in the form of $token\mbox{-} value(p_t)$ $r$ $c$, $Crit(a)$ $::=$ \{$p\vert\forall t$$\in$$^\bullet p_t\colon E(t,p_t)$$\neq$$E(p_t,t)$$\wedge$$p$$\in$$^\bullet t$$\setminus$$P_f$\}.
The slicing criterion w.r.t. $\psi$ is $Crit$ $::=$ \{$p\vert\forall a$$\in$$\mathbb{A}\colon p$$\in$$Crit(a)$\}.
\end{definition}

Intuitively, each proposition in the LTL-$_\mathcal{X}$ formula $\psi$ should extract the places from PDNet $N$ to constitute a slicing criterion w.r.t. $\psi$.
For every proposition $a$ in $\mathbb{A}$, its corresponding places are extracted based on Definition \ref{Def:Crit}.
If $a$ is in the form of $is\mbox{-} fireable(t)$, the input places of $t$ except execution places (i.e., $P_f$ in Definition \ref{Def:PDNet}) are extracted.
If $a$ is in the form of $token\mbox{-} value(p)$ $r$ $c$ where $p$ is a variable place, the transitions in $^\bullet$$p$ that can update the token of $p$ can be found by $E(t,p)$$\neq$$E(p,t)$.
Then, the places except execution places in $^\bullet$$(^\bullet p)$ are extracted.
For example, there is a proposition $is\mbox{-} fireable(t_{21})$ in $\mathcal{F}$ $is\mbox{-} fireable(t_{21})$ of the PDNet in Figure \ref{Fg:Example}($e$). Due to $^\bullet$$t_{21}$$=$$\{f_{21},c_{21}\}$, $c_{21}$ is extracted to $Crit$.

\subsection{On-demand PDNet Slicing Algorithm}\label{Sub:Sli}
As mentioned above, traditional program slicing methods should capture complete data-flow dependencies in advance when constructing PDG.
Differently, we propose a PDNet slicing method to capture data-flow dependencies in Definition \ref{Def:DFD} with an on-demand way.
Let $N$ be the PDNet of concurrent program $\mathcal{P}$, and $Crit$ be a slicing criterion of $N$. We use $\stackrel{d}{\longrightarrow}$ to represent the union of all dependencies of PDNet. We define the PDNet slice next.

\begin{definition}[PDNet Slice]\label{Def:Slice}
Let $N$ be a PDNet, $\psi$ be a LTL-$_\mathcal{X}$ formula, $Crit$ be the slicing criterion w.r.t. $\psi$, and $N'$ be the PDNet slice of $N$ w.r.t. $Crit$.
$N'$ $::$= \{$x\in P\cup T\vert\forall p\in Crit :$ $x\stackrel{d}{\longrightarrow}$$^*$$p$\}.
\end{definition}

In this case, symbol $*$ stands for the potential transitive relationships of $\stackrel{d}{\longrightarrow}$, while $\stackrel{d}{\longrightarrow}$$^*$ denotes the transitive closure of the dependencies of PDNet. This means that regardless of the dependencies for adding a place, the transitive closure is calculated based on the dependencies of $\stackrel{d}{\longrightarrow}$.
When control place $p$ is in $Crit$, we add the control places of transitions that affect $p$ through the constructed arcs that represent control-flow dependencies to $N'$.

As the example in Figure \ref{Fg:Example}($e$), $\{c_{21}\}$ is the slicing criterion for $\mathcal{F}$ $is\mbox{-} fireable(t_{21})$. Due to $t_{2b}$$\stackrel{co}{\longrightarrow}$$t_{20}$ and $t_{20}$$\stackrel{co}{\longrightarrow}$$t_{21}$ according to Definition \ref{Def:CFD}, $t_{2b}$$\stackrel{co}{\longrightarrow}$$^*$$t_{21}$ (i.e., $t_{2b}$ affects $t_{21}$ indirectly), and the control place $c_{2b}$ of $t_{2b}$ should be added to $P'$ of $N'$.
Differently, when a variable place is added to $P'$ of $N'$, its data-flow dependencies can be calculated through the control-flow structure and read-write arcs according to Definition \ref{Def:DFD}.
In Figure \ref{Fg:Example}($e$), when $t_{20}$ is added to $P'$, $v_0$ can be added to $P'$. Then, $t_{10}$$\stackrel{I}{\longrightarrow}$$t_{20}$ concerning the variable place $v_0$ can be captured, and the control place $c_{10}$ of $t_{10}$ is added to $P'$ of $N'$.

Based on the above insights, our on-demand PDNet slicing is realized via algorithm \ref{Alg:Slicing}.

\begin{breakablealgorithm}
\caption{On-demand PDNet Slicing Algorithm}\label{Alg:Slicing}
\begin{algorithmic}[1]
\Require A PDNet $N$ and the slicing criterion $Crit$ w.r.t. $\psi$;
\Ensure The PDNet Slice $N'$ with $P'$ and $T'$; 
    \State $P' := Crit$; /$\ast P'$ is the place set of $N'\ast$/
    \State $T' := \emptyset$; /$\ast T'$ is the transition set of $N'\ast$/
    \State $P_d := \emptyset$; /$\ast P_d$ is a processed place set$\ast$/
    \State \Call {PROPA}{$enter$};
           \Call {PROPA}{$exit$};
\Function{PROPA}{$Dir$} 
    \ForAll{$p$$\in$$(P'$$\cap$$P_c$$\setminus$$P_d)$} 
        \ForAll{$t$$\in$$(^\bullet$$p$$\setminus$$ T')$$\wedge$$(t,p)$$\in$$F_c$} 
            \ForAll{$p'$$\in$$(^\bullet$$t$$\setminus$$P')$$\wedge$$(p',t)$$\in$$F_c$}
                \State $P' := P'$$\cup$$\{p'\}$; 
            \EndFor 
        \EndFor 
        \ForAll{$t$$\in$$p$$^\bullet$$\wedge$$(p,t)$$\in$$F_c$} 
            \State $T' := T'$$\cup$$\{t\}$; 
            \State $P' := P'$$\cup$$(^\bullet$$t$$\cap$$P_f)$; 
            \ForAll{$p'$$\in$($^\bullet$$t$$\setminus$$P'$$\cap$$P_v)$$\wedge$$E(t,p')$$=$$E(p',t)$} 
                \If{ISGLOBAL($p'$)}
                    \State $P' := P'$$\cup$$\{p'\}$; 
                    \ForAll{$t'$$\in$$^\bullet$$p'$$\wedge$$E(t',p')$$\neq$$E(p',t')$} 
                        \If{INPA($t',t,p'$)$\vee$INCON($t',t,p'$)} /$\ast t'$$\stackrel{I}{\longrightarrow}$$t$ or   $t'$$\stackrel{D}{\longrightarrow}$$t\ast$/
                            \ForAll{$p''$$\in$$(^\bullet$$t'$$\setminus$$P')\cap P_c$} 
                                \State $P'$$:=$$P'$$\cup$$\{p''\}$; 
                            \EndFor
                        \EndIf 
                    \EndFor
                \EndIf 
                \If{ISLOCAL($p'$)}
                    \State $T_d :=$ FINDPRE($t,p',Dir$);
                    \If{$T_d$$\neq$$\emptyset$}
                     $P':=P'$$\cup$$\{p'\}$; 
                    \EndIf
                    \ForAll{$t'$$\in$$T_d$} /$\ast t'$$\stackrel{D}{\longrightarrow}$$t\ast$/
                        \ForAll{$p''$$\in$$(^\bullet$$t'$$\setminus$$P')$$\wedge$$(p'',t')$$\in$$F_c$}
                        \State $P'$$:=$$P'$$\cup$$\{p''\}$; 
                        \EndFor
                    \EndFor 
                \EndIf 
            \EndFor
        \EndFor 
        \State $P_d := P_d\cup \{p\}$; 
    \EndFor 
\EndFunction 
\end{algorithmic}
\end{breakablealgorithm}

Firstly, $P'$, $T'$, and $P_d$ are initialized (Lines 1-3). We refer to a place that do not belong to $P_d$ as unprocessed place. That is, $p$$\in$$P_d$ is called a processed place.
There are two calls of PROPA($Dir$) (Line 4) avoiding redundancy caused by multiple calls \cite{Horwitz1990Interprocedural}.
The difference between these two calls is the propagation direction from the call dependence in Definition \ref{Def:CFD}'s $2.1)$ and $2.2)$.

In function PROPA($Dir$), if there exists an unprocessed control place $p$ (Line 6), all nodes $x$ meeting $x$$\stackrel{d}{\longrightarrow}$$^*$$p$ should be captured. The key slicing steps to propagate potential transitive
relationships are as follows.

The first step is to propagate control-flow dependencies from $p$ (Lines 7-9).
A transition $t$$\in$$^\bullet$$p$ is propagated by control arc $(t,p)$ (Line 7). Here, the control arcs describe which statements dominate the current one corresponding to $p$.
According to the control-flow dependencies in Definition \ref{Def:CFD}, $t$$\stackrel{co}{\longrightarrow}$$t'$ or $t$$\stackrel{ca}{\longrightarrow}$$t'$ or $t$$\stackrel{lo}{\longrightarrow}$$t'$ or $t$$\stackrel{po}{\longrightarrow}$$t'$ where $t'$ is the transition of the PDNet structure for the operation where $p$ locates.
Thus, the control place $p'$ of $t$ is added to $P'$ (Lines 8-9).

The second step is to complete the PDNet structure where $p$ locates. Its output transition $t$$\in$$p^\bullet$ is added to $T'$ (Line 11) and the execution place in $^\bullet$$t$ is added to $P'$ (Line 12).

The third step is to propagate data-flow dependencies from $p$ (Lines 13-25).
If there is a variable place $p'$$\in$$Ref(t)$ (Line 13), the data-flow dependencies relevant to $p'$ can be captured based on Definition \ref{Def:DFD}.
As the analysis mentioned above, the interference dependencies are captured only for a global variable to reduce the computation cost.
Thus, if $p'$ corresponds to a global variable (Lines 14-19), $p'$ is added to $P'$ (Line 15). If there exists transition $t'$ such that $p'$$\in$$Def(t')$ (Line 16), function INPA($t',t,p'$) judges whether $t'$$\stackrel{D}{\longrightarrow}$$t$ according to Definition \ref{Def:DFD}'s $1.1)$ and $1.3)$, function INCON($t',t,p'$) judges whether $t'$$\stackrel{I}{\longrightarrow}$$t$ according to Definition \ref{Def:DFD}'s $2.1)$ (Line 17), and the control places of $t'$ are added to $P'$ (Lines 18-19).
If $p'$ corresponds to a local variable (Lines 20-25), function FINDPRE($t,p',Dir$) finds the transition set $T_d::=\{t' \vert p'\in Def(t')$ $\wedge$ $\mathbb{R}(t',t)_{\vec{Dir}}$ $\wedge$ $(\forall t''$$\in$$\mathbb{P}(t',t)$$\colon$$p'$$\notin$$Def(t''))\}$ according to Definition \ref{Def:DFD}'s $1.1)$ and $1.3)$ (Line 21). Here, parameter $Dir$ is $enter$ arc (e.g., $(t_j,c_b)$ in Figure \ref{Fg:Control}($b$)) or $exit$ arc ($(c_e,t_k)$ in Figure \ref{Fg:Control}($b$)).
The control places of $t'$ are added to $P'$ (Lines 23-25).

The final step is that control place $p$ is added to $P_d$ marked as processed (Line 26).

Moreover, $N'$ can be non-executable due to the incomplete execution orders concerning the control-flow structure. Hence, the slicing post-process algorithm is outlined in \cite{Ding2023PDNet}.
The complexity of Algorithm \ref{Alg:Slicing} is, in the best case, $O$($n$), where $n$ is the statement count meaning the number of statements. Its worst-case complexity is $O$($n^2$) when the algorithm cannot slice any statement.
The correctness is proved in \cite{Ding2023PDNet}.

\begin{figure*}[t]\centering
  \includegraphics[width=\textwidth]{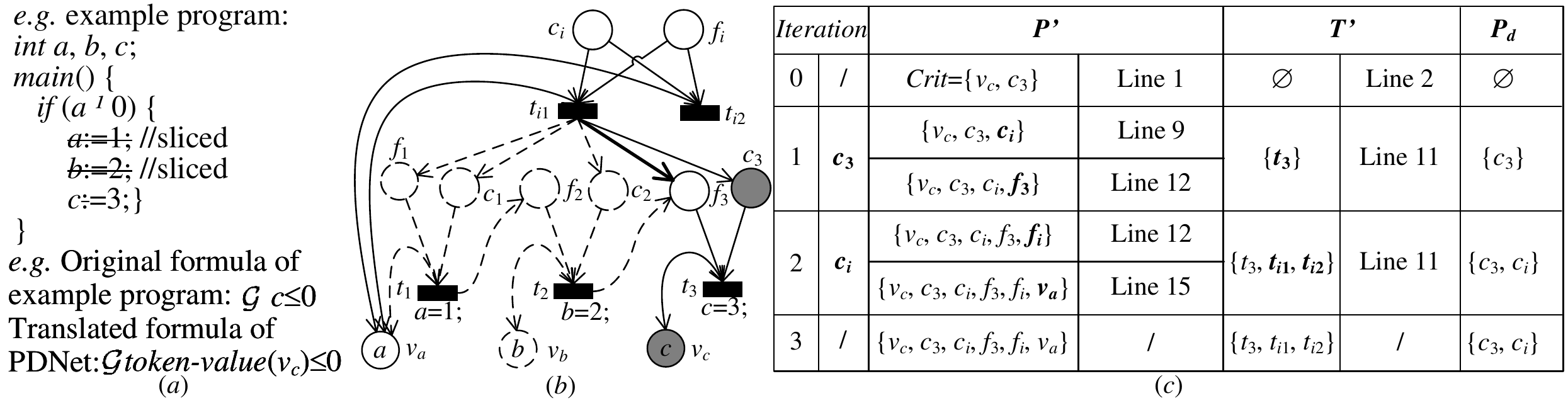}
  \caption{Example for slicing process ($a$) The property and the example program ($b$) The marked PDNet slice ($c$) The slicing process by Algorithm \ref{Alg:Slicing}}\label{Fg:Slice}
\end{figure*}

In Figure \ref{Fg:Slice}, we set $\mathcal{G}$ $token\mbox{-} value(v_c)\leq 0$ as the example property. $a$$:=$$1$ and $b$$:=$$2$ can be sliced for this property in Figure \ref{Fg:Slice}($a$). The PDNet of the example program is in Figure \ref{Fg:Slice}($b$).

Firstly, we extract $Crit$ through Definition \ref{Def:Crit}. There is only a proposition $token\mbox{-} value(v_c)\leq 0$, where variable place $v_c$ corresponds to variable $c$.
Due to $E(t_3,v_c)$$\neq$$E(v_c,t_3)$, $^\bullet$$t_3$$\setminus$$P_f$ is extracted, and $Crit$$=$$\{v_3,c_3\}$ is filled with dark gray.

Then, we update $P'$ and $T'$ via Algorithm \ref{Alg:Slicing}.
In the table of Figure \ref{Fg:Slice}($c$), $P'$ and $T'$ are updated in each iteration. Initially, $P'$$=$$\{v_c,c_3\}$, $T'$$=$$\emptyset$, and $P_d$$=$$\emptyset$ by Algorithm \ref{Alg:Slicing}.
\begin{itemize}
    \item In the first iteration, $c_3$ is selected as the control place by Algorithm \ref{Alg:Slicing}. Control place $c_i$ is added to $P'$ according to control dependence in Definition \ref{Def:CFD}. Transition $t_3$ is added to $T'$, and $f_3$ is added to $P'$ by Algorithm \ref{Alg:Slicing}. $c_3$ is marked in $P_d$ as a processed place by Algorithm \ref{Alg:Slicing}.
    \item In the second iteration, $c_i$ is selected in $P'$$\cap$$P_c$$\setminus$$P_d$. $t_{i1}$ ($t_{i2}$ with the same process) is added to $T'$ and execution place $f_i$ is added to $P'$ by Algorithm \ref{Alg:Slicing}. For data-flow dependencies in Definition \ref{Def:DFD}, variable place $v_a$ is selected due to $E(t_{i1},v_a)$$=$$E(v_a,t_{i1})$. $a$ corresponding to $v_a$ is a global variable, and $v_a$ is added to $P'$ by Algorithm \ref{Alg:Slicing}. Then, transition $t_a$ is selected due to $E(t_{a},v_a)$$\neq$$ E(v_a,t_{a})$. But there exists no execution path by function INPA in Algorithm \ref{Alg:Slicing}, and no control place is added to $P'$. $c_i$ is marked in $P_d$ as a processed place.
    \item Finally, no control place is selected based on $P'$ and $P_d$.
\end{itemize}

The dotted places and dotted arcs, as well as $t_1$ and $t_2$, are removed in Figure \ref{Fg:Slice}($b$).
However, there is no transition $t$ in $^\bullet f_3$ such that $(t,f_3)$ is an execution arc. Thus, a new execution arc $(t_{i1},f_3)$ expressed as a bold arrow is constructed based on the post-process.

\section{Experimental Evaluation}\label{Sec:Exp}
\subsection{Tool Implementation and Benchmarks}\label{Sub:Imp}
We implement a model checking tool called \textit{DAMER} (Dependence Analyser and Multi-threaded programs checkER) \cite{Chen2021Automatic}.
The input programs are automatically translated into a PDNet, and our on-demand PDNet slicing method is used to reduce it. The resulting PDNet slice can then be checked by using explicit-state model checking.
\textit{DAMER} supports the features of C programs with POSIX threads \cite{PThread2019} and offers support for various linear temporal properties. These properties can be expressed as LTL-$_\mathcal{X}$ formulae, as defined in Definition \ref{Def:LTL}. Once modeled, the relevant program variables or locations can be translated into PDNet propositions within \textit{DAMER} automatically.

To showcase the practical effectiveness of our methods in verifying LTL, we evaluated on a set of multi-threaded C programs, randomly chosen from the Software Verification Competition \cite{SV} and benchmarks in \cite{Li2023CAMC}.
Our experiment involves five concurrent programs that do not use POSIX thread functions. \textit{Fib} is a mathematical algorithm that divides tasks into multiple threads and combines the constraints of each computation. The other four algorithms we use to solve mutual exclusion problems in concurrent systems are \textit{Lamport} \cite{Ben2014Beyond}, \textit{Dekker} \cite{Burnim2011Testing}, \textit{Szymanski} \cite{Guo2016Conc}, and \textit{Peterson} \cite{Rodriguez2015Unfolding}.
Additionally, it involves five concurrent programs that use POSIX thread functions. \textit{Sync} \cite{SV} is an example program that implements thread synchronization through a condition variable. The remaining four programs i.e., \textit{Datarace}, \textit{Varmutex}, \textit{Rwlock}, and \textit{Lazy} \cite{SV}, all access shared memory protected by a single mutex lock.
For each program, we set two formulae to specify their safety properties expressed by a specific $error$ location and constraint properties expressed by the constraints relevant to the key variables.
Our code is publically available \footnote{\url{https://github.com/shirleylee03/damer/}\label{tool}}.

The experiments are conducted with $16GB$ memory. For the used benchmarks, the variation in time among different runs is relatively small. Therefore, it is sufficient to run them $10$ times to use the average.
The original complete result, as well as verified formulae ($\psi_1$ and $\psi_2$ in Tables \ref{Tab:SliCom} and \ref{Tab:VeriCom}), are publically available \ref{tool} for readers to examine. Here, all reported runtime is in milliseconds.
The whole process of each verified property for each benchmark is called verification in the following.

\subsection{PDNet v.s. its Peers}\label{Sub:Met}
We compare the performance of our method, called \textit{PDNet}, with slicing concurrent programs for model checking, called \textit{PS}, and slicing Petri nets, called \textit{TraPNSlice}.
\textit{PS} first builds PDG based on traditional control-flow dependencies and data-flow dependencies according to the definitions and calculation methods \cite{Qi2017Precise} to reduce the concurrent program itself. Then, the program slice is transformed to a traditional CPN model \cite{Jensen2009Formal} (as the example in Figure \ref{Fg:Example}($d$)) and the properties are verified by the same model checking algorithm. In this way, there exist the model conversions mentioned above, implemented on the same tool \textit{DAMER}.
\textit{TraPNSlice} first builds a traditional CPN model for the concurrent program and then slices this model for reduction. Also, such a reduced model is verified by the same model checking algorithm, implemented on the same tool \textit{DAMER}.
We evaluate the efficiency and effectiveness of our method on \textit{DMAER} by answering the following questions:

$Q1$: Are the unified model and on-demand slicing method based on PDNet more efficient than its peers?

$Q2$: Is our verification method based on our PDNet slice more effective than its peers?

\subsection{Slice Comparison}\label{Sub:SliceCom}

\begin{table*}[t]\centering
\caption{Slice Comparison of \textit{PS} and \textit{PDNet} (in milliseconds)}\label{Tab:SliCom}
\resizebox{\textwidth}{!}{
\begin{tabular}{@{}cccccccccccc@{}}
\toprule[1pt]

\multirow{2}{*}{Bench.} & \multirow{2}{*}{For.} & \multicolumn{6}{c}{\textit{PS}} & \multicolumn{3}{c}{\textit{PDNet}} & $S_{com}$/ \\ \cmidrule{3-11}

\multicolumn{2}{c}{} & $S_{dfd}$ & $S_{cfd}$ & $S_{clo}$ & $S_{ps}$ & $S_{mol}$ & $S_{com}$ & $T_{mol}$ & $T_{sli}$ & $T_{com}$ & $T_{com}$ \\ \midrule

\multirow{2}{*}{Fib}
& $\psi_1$    & 1.334     & 18.183  & 0.301     & 19.818   & 1.578     & 21.396  & 2.240     & 0.744   & 2.984 & 7.170  \\
& $\psi_2$ & 1.280  & 17.510  & 0.331 & 19.121   & 0.931 & 20.052  & 2.298 & 0.324 & 2.622 & 7.648 \\

\multirow{2}{*}{Lamport}
& $\psi_1$  & 2.006 & 165.664 & 0.578 & 168.248  & 2.552 & 170.800   & 4.201 & 1.610 & 5.811 & 29.393 \\
& $\psi_2$  & 1.968 & 165.432 & 0.854  & 168.254  & 2.257    & 170.511  & 4.209  & 1.504  & 5.712 & 29.850  \\

\multirow{2}{*}{Dekker}
& $\psi_1$ & 1.428  & 64.792  & 0.430 & 66.650   & 1.949  & 68.599    & 3.213  & 1.286 & 4.499 & 15.249 \\
& $\psi_2$  & 1.443  & 64.299  & 0.535  & 66.277   & 1.520  & 67.797 & 3.196 & 1.033  & 4.230 & 16.029 \\

\multirow{2}{*}{Szymanski}
& $\psi_1$ & 1.745  & 449.599 & 0.966     & 452.311  & 2.608  & 454.918   & 4.115     & 1.439  & 5.554  & 81.909  \\
& $\psi_2$  & 1.750  & 448.491 & 1.174 & 451.415  & 1.982  & 453.397  & 4.149  & 1.363  & 5.512 & 82.256 \\

\multirow{2}{*}{Peterson}
& $\psi_1$ & 1.282  & 40.541  & 0.252 & 42.075   & 1.577 & 43.652 & 2.675  & 1.109  & 3.783 & 11.538 \\
& $\psi_2$ & 1.317  & 40.171  & 0.340     & 41.828   & 1.200 & 43.028 & 2.585  & 0.836 & 3.421 & 12.576 \\

\multirow{2}{*}{Sync}
& $\psi_1$ & 1.757 & 14.505  & 0.342 & 16.605   & 1.265 & 17.870 & 1.838 & 0.743 & 2.581 & 6.923 \\
& $\psi_2$ & 1.725  & 14.545  & 0.679  & 16.949   & 1.146  & 18.095  & 1.817  & 0.728  & 2.545 & 7.110 \\

\multirow{2}{*}{Datarace}
& $\psi_1$ & 1.869 & 51.096  & 0.415 & 53.380 & 2.056 & 55.436 & 3.114 & 0.962 & 4.076 & 13.599  \\
& $\psi_2$ & 1.886  & 51.024  & 0.782  & 53.693   & 1.872 & 55.564 & 3.145  & 1.010 & 4.155 & 13.372 \\

\multirow{2}{*}{Rwlock}
& $\psi_1$ & 1.836 & 32.237  & 0.869 & 34.942   & 2.304  & 37.246  & 3.022 & 1.475  & 4.497 & 8.283 \\
& $\psi_2$ & 1.847 & 31.838  & 0.934 & 34.619   & 1.895  & 36.514  & 3.051 & 1.393  & 4.444 & 8.217 \\

\multirow{2}{*}{Varmutex}
& $\psi_1$  & 1.730 & 29.709  & 0.675 & 32.114   & 3.464 & 35.578 & 4.598  & 1.648  & 6.245 & 5.697 \\
& $\psi_2$  & 1.703  & 29.633  & 0.702 & 32.038 & 3.044 & 35.082 & 4.690 & 2.037 & 6.727 & 5.215 \\

\multirow{2}{*}{Lazy}
& $\psi_1$ & 1.620 & 11.914  & 0.232   & 13.766   & 1.356  & 15.122  & 1.784  & 0.754  & 2.538 & 5.959 \\
& $\psi_2$  & 1.662 & 12.014  & 0.421 & 14.097   & 1.035 & 15.132  & 1.785 & 0.733 & 2.518 & 6.009 \\ \midrule

\multicolumn{2}{c}{Average} & 1.659  & 87.660 & 0.591  & 89.910 & 1.879  & 91.789 & 3.086  & 1.137 & 4.223 & 18.700 \\
\bottomrule[1pt]
\end{tabular}}
\end{table*}

For $Q1$, we report the concrete slicing time.
In Table \ref{Tab:SliCom}, $S_{dfd}$ is the time to calculate the complete data-flow dependencies for PDG,  $S_{cfd}$ is the time to calculate the complete control-flow dependencies for PDG, and $S_{clo}$ is the time to capture the transitive closure of PDG.
Then, $S_{ps}$$=$$S_{dfd}$$+$$S_{cfd}$$+$$S_{clo}$ is the slicing time of \textit{PS}.
$S_{mol}$ is the modeling time to construct a CPN model for the program slice, and $S_{com}$$=$$S_{ps}$$+$$S_{mol}$ is the whole time of \textit{PS} before model checking.
As for \textit{PDNet}, $T_{mol}$ is the modeling time to construct the PDNet for the whole program, and $T_{sli}$ is the slicing time.
Here, $T_{mol}$ includes the time to calculate the control-flow structure and control-flow dependencies, and $T_{sli}$ includes the time to capture the data-flow dependencies as well as the transitive closure of PDNet.
$T_{com}$$=$$T_{mol}$$+$$T_{sli}$ is the whole time of \textit{PDNet} before model checking. Their average values are on the last row.

As we can see from Table \ref{Tab:SliCom}, $S_{com}$$>$$T_{com}$ holds for each verification.
The average of $S_{com}$ is $91.789$, and the average of $T_{com}$ is $4.222$. It implies the higher efficiency of our proposed method.
We calculate $S_{com}/T_{com}$ in Table \ref{Tab:SliCom}, and its average is $18.7$. Here, $10$ verification reduces the whole time by more than $10$ times, and $4$ verification reduces the whole time by more than $30$ times.
Our method can reduce computation when converting among multiple models.
Moreover, $S_{dfd}$$+$$S_{clo}$ of \textit{PS}, including the computation time of the data-flow dependencies and the transitive closure, corresponds to $T_{sli}$ of \textit{PDNet}.
It is obvious that $S_{dfd}$$+$$S_{clo}$$>$$T_{sli}$ holds for each verification.
The variable's data-flow dependencies are only calculated when it is added to the PDNet slice. This means that the use of PDNet can save on computation cost by avoiding the computation of irrelevant variable data-flow dependencies.

Note that $S_{mol}$$<$$T_{mol}$ holds for each verification. The average of $S_{mol}$ is $1.879$, and that of $T_{mol}$ is $3.086$. Their difference is $1.207$ which is insignificant.
Our approach to utilizing the unified model and on-demand slicing based on PDNet proves to be highly efficient, as it significantly reduces computation during model conversions and eliminates data-flow dependencies of irrelevant variables.
CPN is translated from a program slice; while PDNet is translated from the entire program. Thus PDNet needs extra places and arcs to depict control-flow dependencies. Yet such extra ones lead to the small impact on the overall time. In summary, our approach proves to be a much more cost-effective solution than its peers.

\subsection{Verification Comparison}\label{Sub:VeriCom}

\begin{table*}[t]\centering
\caption{Verification Comparison of \textit{PS}, \textit{TraPNSlice} and \textit{PDNet} (in microseconds)}\label{Tab:VeriCom}
\resizebox{\textwidth}{!}{
\begin{tabular}{@{}cccccccccccc@{}}
\toprule[1pt]
\multirow{2}{*}{Bench.} &
  \multirow{2}{*}{For.} &
  \multicolumn{2}{c}{\textit{PS}} &
  \multicolumn{4}{c}{\textit{TraPNSlice}} &
  \multicolumn{2}{c}{\textit{PDNet}} &
  \multirow{2}{*}{\begin{tabular}[c]{@{}c@{}}$S_{ver}$/\\ $T_{ver}$\end{tabular}} &
  \multirow{2}{*}{\begin{tabular}[c]{@{}c@{}}$C_{ver}$/\\ $T_{ver}$\end{tabular}} \\ \cmidrule{3-10}
  &  & $S_{ver}$ & $S_{res}$ & $C_{mol}$ & $C_{sli}$ & $C_{ver}$ & $C_{res}$ & $T_{ver}$ & $T_{res}$ &      &      \\ \midrule
\multirow{2}{*}{Fib}       & $\psi_1$ & 14931.152 & $F$       & 3.505     & 0.201     & 28951.807 & $F$       & 10544.388 & $F$       & 1.416 & 2.746 \\
                           & $\psi_2$ & 9.686     & $F$       & 3.501     & 0.152     & 18.267    & $F$       & 7.511     & $F$       & 1.290 & 2.432 \\
\multirow{2}{*}{Lamport}   & $\psi_1$ & 38.590    & $T$       & 5.917     & 0.617     & 77.759    & $T$       & 35.822    & $T$       & 1.077 & 2.171 \\
                           & $\psi_2$ & 8.075     & $F$       & 5.937     & 0.631     & 13.454    & $F$       & 7.339     & $F$       & 1.100 & 1.833 \\
\multirow{2}{*}{Dekker}    & $\psi_1$ & 13.782    & $T$       & 4.658     & 0.348     & 22.719    & $T$       & 11.251    & $T$       & 1.225 & 2.019 \\
                           & $\psi_2$ & 12.810    & $T$       & 4.637     & 0.369     & 23.224    & $T$       & 11.489    & $T$       & 1.115 & 2.021 \\
\multirow{2}{*}{Szymanski} & $\psi_1$ & 27.037    & $T$       & 5.778     & 0.534     & 44.302    & $T$       & 20.703    & $T$       & 1.306 & 2.140 \\
                           & $\psi_2$ & 9.968     & $F$       & 5.851     & 0.534     & 17.827    & $F$       & 9.666     & $F$       & 1.031 & 1.844 \\
\multirow{2}{*}{Peterson}  & $\psi_1$ & 12.695    & $T$       & 3.900     & 0.282     & 25.197    & $T$       & 11.344    & $T$       & 1.119 & 2.221 \\
                           & $\psi_2$ & 11.689    & $T$       & 3.889     & 0.290     & 26.034    & $T$       & 11.634    & $T$       & 1.005 & 2.238 \\
\multirow{2}{*}{Sync}      & $\psi_1$ & 8.420     & $T$       & 2.726     & 0.248     & 17.253    & $T$       & 8.415     & $T$       & 1.001 & 2.050 \\
                           & $\psi_2$ & 8.657     & $T$       & 2.735     & 0.230     & 17.584    & $T$       & 8.346     & $T$       & 1.037 & 2.107 \\
\multirow{2}{*}{Datarace}  & $\psi_1$ & 8009.776  & $T$       & 4.744     & 0.361     & 15999.411 & $T$       & 6947.670  & $T$       & 1.153 & 2.303 \\
                           & $\psi_2$ & 17.369    & $F$       & 4.768     & 0.350     & 31.782    & $F$       & 15.225    & $F$       & 1.141 & 2.088 \\
\multirow{2}{*}{Rwlock}    & $\psi_1$ & 40.028    & $T$       & 5.183     & 0.466     & 70.768    & $T$       & 33.947    & $T$       & 1.179 & 2.085 \\
                           & $\psi_2$ & 52.972    & $T$       & 5.143     & 0.466     & 72.940    & $T$       & 34.491    & $T$       & 1.536 & 2.115 \\
\multirow{2}{*}{Varmutex}  & $\psi_1$ & 45.990    & $T$       & 8.657     & 0.546     & 79.711    & $T$       & 27.537    & $T$       & 1.670 & 2.895 \\
                           & $\psi_2$ & 82.718    & $T$       & 8.485     & 0.534     & 81.173    & $T$       & 30.926    & $T$       & 2.675 & 2.625 \\
\multirow{2}{*}{Lazy}      & $\psi_1$ & 11.292    & $F$       & 2.798     & 0.193     & 18.521    & $F$       & 8.857     & $F$       & 1.275 & 2.091 \\
                           & $\psi_2$ & 8.074     & $F$       & 2.801     & 0.181     & 13.325    & $F$       & 7.465     & $F$       & 1.082 & 1.785 \\ \midrule
\multicolumn{2}{c}{Average}           & 1168.039  & /         & 4.781     & 0.377     & 2281.153  & /         & 889.701   & /         & 1.272 & 2.190  \\
\bottomrule[1pt]
\end{tabular}}
\end{table*}

For $Q2$, we report the verifying time and results ($T$ represents $true$, while $F$ represents $false$).
In Table \ref{Tab:VeriCom}, $S_{ver}$ is the time for verifying the CPN model of the program slice, and $S_{res}$ is the verifying result of \textit{PS}.
As for \textit{TraPNSlice}, $C_{mol}$ is the modeling time to construct a CPN model for the whole program, $C_{sli}$ is the time of traditional Petri net slicing, $C_{ver}$ is the time for verifying traditional Petri net slice, and $C_{res}$ is the verifying result.
$T_{ver}$ is the time for verifying PDNet slice, and $T_{res}$ is the verifying result of \textit{PDNet}.

As we can see from Table \ref{Tab:VeriCom}, $S_{res}$, $C_{res}$, and $T_{res}$ are all correct. It implies that the three slicing methods are correct and effective in \textit{DAMER}.
We calculate $S_{ver}/T_{ver}$ and $C_{ver}/T_{ver}$ as shown in Table \ref{Tab:VeriCom}.
Obviously, $S_{ver}$$>$$T_{ver}$ and $C_{ver}$$>$$T_{ver}$ holds for each verification, implying that our PDNet slicing with an on-demand way outperforms the program slicing and traditional Petri net slicing methods in terms of verification time.
The average value of $S_{ver}/T_{ver}$ is $1.272$.
As for the traditional program slicing method, it takes much time to calculate the domination relationship for the control-flow dependencies \cite{Qi2017Precise}.
The average value of $C_{ver}/T_{ver}$ is $2.190$.
It is evident that our PDNet slicing method can reduce the verifying time more than traditional Petri net slicing. The reason is that it can reduce the number of places, transitions, and explored states of the original PDNet, which is not reduced by traditional Petri net slicing.

Furthermore, the time required for modeling ($T_{mol}$ in Table \ref{Tab:SliCom}) and slicing ($T_{sli}$ in Table \ref{Tab:SliCom}) are significantly less than $T_{ver}$ in Table \ref{Tab:VeriCom}. As a result, the expense of constructing a PDNet model is reasonable and the reduction achieved through PDNet slicing is outstanding. This confirms the effectiveness and practicality of PDNet slicing.

\section{Related Work}\label{Sec:Rel}

Rakow \cite{Rakow2012Safety} suggested two static slicing algorithms to get a simplified Petri net with preserving CTL*-$_\mathcal{X}$ properties.
Khan \cite{Khan2013OptimizingVer} improved Rakow's algorithms and suggested a dynamic slicing algorithm for Algebraic Petri nets (APN).
Lorens et al. \cite{Llorens2008Dynamic,Llorens2016Dynamicslicing} and Yu et al. \cite{Yu2013Extended,Yu2015Dynamic} proposed two dynamic slicing algorithms that took into account the initial marking of Petri nets, further reducing the scale of the Petri net slice.
Then, two algorithms are improved by \cite{Llorens2023Maximal} that encounters the maximal slice and the minimal one from the initial net slice.
Roci et al. \cite{Roci2020A} proposed a restraining algorithm to slice from the Petri Nets model the nodes that participate in some particular executions.
Wang et al. \cite{Wang2020Measurement} proposed an improved Dynamic-Slicing-based Vulnerability Detection (IDS-VR) method to locate vulnerabilities in Workflow Nets. They \cite{Wang2020A} then proposed a Dynamic-Data-Slicing-based Vulnerability Detection (DDS-VD) method for data inconsistency detection of E-Commerce Systems.
These Petri net slicing algorithms \cite{Khan2018PNSlicing} are primarily suitable for workflow. The dependencies outlined in this paper are more complex and pose a challenge to the existing methods \cite{Khan2018PNSlicing,Roci2020A,Wang2020Measurement,Wang2020A,Llorens2023Maximal}.

Danicic et al. \cite{Danicic2018Static} defined non-termination insensitive (weak) slices and non-termination sensitive (strong) slices for non-deterministic programs. Chalupa et al. \cite{Chalupa2021Fast} proposed new algorithms for computing non-termination sensitive control dependence (NTSCD) and decisive order dependence (DOD) for fast Computation.
Masud et al. \cite{Masud2021Semantic} contributed a general proof of correctness for dependence-based slicing methods for interprocedural, possibly nonterminating programs. Although these works provide more precise definitions, semantics and proof of dependencies for non terminating programs, they are used for sequential programs, instead of concurrent program.
\section{Conclusion}\label{Sec:Con}
This paper introduces PDNet as a unified model for representing control-flow structures with dependencies. This saves computation by eliminating the need to convert among multiple models. Then, we propose an on-demand PDNet slicing method that reduces the scope of model checking by capturing data-flow dependencies related to variables from verified LTL-$_\mathcal{X}$. Our methodology is able to save cost from model conversions and complete calculation of data-flow dependencies, which is never seen in the existing work. We also implement an automatic concurrent program model checking tool called \textit{DAMER} based on PDNet for LTL-$_\mathcal{X}$ formulae. Our experiments have produced promising results.

Based on the current work, we plan to add the heuristic information from the dependencies information to help find counterexample paths for LTL-$_\mathcal{X}$ formulae as quickly as possible during exploration and integrate partial-order methods, e.g., partial order reduction or unfolding \cite{Liu2016A,Dou2019Maximal}, to reduce possible interleavings from the independent threads for PDNet.

\section*{Acknowledgments}\label{Sec:Ack}
\textbf{Funding} This work was supported by National Key Research and Development Program of China under Grant 2022YFB4501704.

\textbf{Statement} 
The authors have no relevant financial or non-financial interests, no potential competing interests and no applicable ethics approval to disclose. 


\printcredits

\bibliographystyle{cas-model2-names}
\bibliography{cas-refs-PDNet}

\newpage
\appendix


\section{Colored Petri Nets}\label{App:CPN}
To define PDNet (Program Dependence Net) based on CPN, we introduce the definitions of multiset and CPNs\cite{Jensen2009Formal}.
\begin{definition}[Multiset]
Let $S$ be a non-empty set. A multiset $ms$ over $S$ is a function $ms\colon S\rightarrow\mathbb{N}$ that maps each element into a non-negative integer.
$S_{MS}$ is the set of all multisets over $S$.
We use $+$ and $-$ for the sum and difference of two multisets.
And $=$, $<$, $>$, $\leq$, $\geq$ are comparisons of multisets, which are defined in the standard way.
\end{definition}

\begin{definition}[Colored Petri Net]\label{Def:CPN}
CPN is defined by a $9$-tuple $N$ $::=$ $(\Sigma$, $V$, $P$, $T$, $F$, $C$, $G$, $E$, $I)$, where:

1. $\Sigma$ is a finite non-empty set of types called color sets.

2. $V$ is a finite set of the typed variables. $\forall v$$\in$$V\colon\\Type[v]$$\in$$\Sigma$.

3. $P$ is a finite set of places.

4. $T$ is a finite set of transitions and $T\cap P=\emptyset$.

5. $F\subseteq(P$$\times$$T)\cup(T$$\times$$P)$ is a finite set of directed arcs.

6. $C\colon P$$\rightarrow$$\Sigma$ is a color set function, that assigns a color set $C(p)$ belonging to the set of types $\Sigma$ to each place $p$.

7. $G\colon T$$\rightarrow$$\mathbb{E}_V$ is a guard function, that assigns an expression $G(t)$ to each transition $t$. $\forall t$$\in$$T\colon(Type[G(t)]\in BOOL)\wedge(Type[Var(G(t))]\subseteq\Sigma)$.

8. $E\colon F$$\rightarrow$$\mathbb{E}_V$ is a function, that assigns an arc expression $E(f)$ to each arc $f$. $\forall f$$\in$$F\colon(Type[E(f)]=C(p(f))_{MS})$ $\wedge$ $(Type[Var(E(f))]\subseteq\Sigma)$, where $p(f)$ is the place connected to arc $f$.

9. $I\colon P$$\rightarrow$$\mathbb{E}_\emptyset$ is an initialization function, that assigns an initialization expression $I(p)$ to each place $p$. $\forall p$$\in$$P\colon\\(Type[I(p)]=C(p)_{MS})$$\wedge$$(Var(I(p))=\emptyset)$.
\end{definition}

The difference between PDNet and CPNs as the following aspects:
1. $P$ is divided into three subsets in PDNet, i.e., $P=P_c\cup P_v\cup P_f$. Concretely, $P_c$ is a subset of control places, $P_v$ is a subset of variable places, and $P_f$ is a subset of execution places.
2. $F$ is divided into three subsets in PDNet, i.e., $F=F_c\cup F_{rw}\cup F_f$. Concretely, $F_c\subseteq(P_c$$\times$$T) \cup (T$$\times$$P_c)$ is a subset of control arcs, $F_{rw}\subseteq(P_v$$\times$$T) \cup (T$$\times$$P_v)$ is a subset of read-write arcs, and $F_f\subseteq(P_f$$\times$$T) \cup (T$$\times$$P_f)$ is a subset of execution arcs.
In addition to the above differences, other definitions and constraints of PDNet are consistent with CPNs.

\section{Concurrent Program Semantics}\label{App:LTS}
C programs using POSIX threads \cite{PThread2019} refer to the concurrent programs in this paper. 
For simplicity, we consider the assignment statements to be atomic.
Take inspiration from the existing research on the function call \cite{Horwitz1990Interprocedural,Masud2021Semantic} and concurrency primitive \cite{PThread2019}, we introduce a simple concurrent program definition.
Suppose that $\mathcal{V}$ is a set of basic program variables (e.g., the types of $int$ and $double$), $val$ is a set of all values that a variable $\nu$ in $\mathcal{V}$ can take.
Suppose that $\mathcal{I}$ is a set of thread identifiers (i.e., \emph{pthread\_t}), $W$ is a set of program expressions, $Q$ is a set of operations that characterize the nature of the action performed by the statement.

\begin{definition}[Concurrent Program]\label{Def:CP}
$\mathcal{P}$ $::=$ $\langle$$K$, $\mathcal{M}$, $\mathcal{L}$, $\mathcal{C}$, $\mathcal{T}$, $\mathcal{H}$, $\mathcal{R}$, $m_0$, $h_0$$\rangle$ is a concurrent program, where:

1. $K$ is a finite set of all program locations.

2. $\mathcal{M}$ is a set of all memory states.

3. $\mathcal{L}$ is a finite set of POSIX mutex variables (i.e., pthread\_mutex\_t). $\ell\in\mathcal{L}$ is a mutex.

4. $\mathcal{C}$ is a finite set of condition variables (i.e., pthread\_cond\_t). $\gamma\in\mathcal{C}$ is a condition variable.

5. $\mathcal{T}\subseteq\mathcal{I}\times Q\times K\times K\times \mathcal{M}\times\mathcal{M}$ is a finite set of statements.

6. $\mathcal{H}\colon\mathcal{I}$$\rightarrow$$K$ is a function that assigns its current location to each thread identifier. 

7. $\mathcal{R}\colon\mathcal{M}$$\rightarrow$$val_{MS}$ is a function that assigns the current values of variables to each memory state. 

8. $h_0\in \mathcal{H}$ is the initial location function that assigns the initial location to each thread identifier.

9. $m_0\in\mathcal{M}$ is an initial memory state.
\end{definition}

\begin{remark}
  A statement $\tau::=\langle i, q, l, l', m, m'\rangle$ intuitively represents that a thread $i$$\in$$\mathcal{I}$ can execute an operation $q$$\in$$Q$, updating the location from $l$$\in$$ K$ to $l'$$\in$$ K$ and the memory state from $m$$\in$$\mathcal{M}$ to $m'$$\in$$\mathcal{M}$.
\end{remark}

The corresponding syntax of concurrent program $\mathcal{P}$ in this paper is described in Table \ref{Tab:syn}, where $\epsilon$, $val$, $\nu$, $\gamma$, $\ell$, $uop$, $rop$, $break$, $continue$, $return$, $entry$, $export$, $i$, $if$, $then$, $else$, $while$, $do$, $call$, $rets$, $lock$, $unlock$, $wait$, and $signal$ are the terminal symbols of a syntax.
Here, $\epsilon$ means the default value, $uop$ is a set of unary operators, $rop$ is a set of binary operators.
$\mathcal{P}$ contains a series of variable declarations $\nu^*$ and function declarations $fun^+$. 
A function $fun$ is uniquely identified by $i$, and contains an $entry$ and an $exit$, and $\nu^*$ represents a parameter list that can be defaulted. $\tau^*$ is a set of statements within this function.

The behavior of a statement $\tau$$\in$$\mathcal{T}$ is represented by its operation $q$, characterizing the nature of the action performed by this statement.
Then, we distinguish the following operation sets $\{local\}$, $\{calls\}$ and $\{syncs\}$ for $\tau ::= local$\textbar$calls$\textbar$syncs$ in Table \ref{Tab:syn}.
Here, operations assignment, jump, and branching belong to $\{local\}$, operations call site and return site belong to $\{calls\}$, and $\langle lock, \ell\rangle$, $\langle unlock, \ell\rangle$, $\langle signal, \gamma\rangle$ and $\langle wait, \gamma, \ell\rangle$ belong to $\{syncs\}$.

($1$) The local operations in $local$ can be used to model statements within a local thread. $\nu$$:=$$w$ represents a simple assignment operation where $\nu$$\in$$\mathcal{V}$ is a program variable and $w$$\in$$W$ is an expression over the program variables.
$jump$ represents a simple jump operation with a particular symbol, e.g., $break$, $continue$ and $return$.
$if(w)$ $then$ $(\tau_1^*)$ $else$ $(\tau_2^*)$ (abbreviated as $if(w)$) is a branch conditional structure with a boolean condition denoted by an expression $w$, and $while(w)$ $do$ $(\tau^*)$ (abbreviated as $while(w)$) is a loop structure with a boolean condition denoted by an expression $w$. The two structures are the branching operations that produce different possible subsequent executions.

($2$) $calls$ represents possible many function calls. As the syntax in Table \ref{Tab:syn}, $call$ represents call site operation making the control-flow turn to the called function, and $rets$ represents return site one, which makes the control-flow turn back from the called function.
Moreover, $cassign$ represents the assignments to all formal input parameters of $call$, and $rassign$ represents the assignments to the actual return parameters of $rets$.

($3$) The POSIX thread operations in $syncs$ could model the synchronization statements in different threads. $syncs=(\{lock,unlock\}$$\times$$\mathcal{L})\cup(\{signal\}$$\times$$ \mathcal{C}))\cup(\{wait\}$$\times$$\mathcal{C}$$\times$$\mathcal{L})$.
Operation $\langle lock, \ell\rangle$ represents a request operation to acquire $\ell$$\in$$\mathcal{L}$, i.e., pthread\_mutex\\\_lock(\verb'&'$\ell$), while $\langle unlock, \ell\rangle$ represents a request operation to release $\ell$$\in$$\mathcal{L}$, i.e., pthread\_mutex\_unlock(\verb'&'$\ell$).
$\langle signal, \gamma\rangle$ represents a request operation to signal other thread on $\gamma$$\in$$\mathcal{C}$, i.e., pthread\_cond\_signal(\verb'&'$\gamma$).
$\langle wait, \gamma, \ell\rangle$ represents a request operation to wait for a notification on $\gamma$$\in$$\mathcal{C}$ with $\ell$$\in$$\mathcal{L}$, i.e., pthread\_cond\_wait(\verb'&'$\gamma$, \verb'&'$\ell$).
Concrete arguments of POSIX thread functions mentioned above can be found from \cite{PThread2019}.

\begin{table}[h]
\begin{center} 
\caption{Simplified Syntax of Concurrent Programs}\label{Tab:syn}
\resizebox{0.5\textwidth}{!}{
\begin{tabular}{l}
\toprule
    $\mathcal{P} ::= (\nu^*) fun^+$ \\
    $fun ::= entry$ $i$ $(\nu^*)$ $(\tau^*)$ $exit$ \\
    $\tau ::= local$\textbar$calls$\textbar$syncs$ \\
    $local ::= \nu$$:=$$w$\textbar $jump$\textbar$if(w)$$then$$(\tau_1^*)$$else$$(\tau_2^*)$\textbar$while(w)$$do$$(\tau^*)$ \\
    $jump ::= break$\textbar$continue$\textbar $return$ \\
    $w ::= val$\textbar$\nu$\textbar$uop$ $w$\textbar$w$ $rop$ $w$ \\
    $calls ::= acall$\textbar$cassign$ $acall$ $rassign$ \\
    $acall ::= call$ $cassign$ $acall$ $rassign$ $rets$\textbar$acall$ $acall$ \textbar $\epsilon$ \\
    $cassign ::= (\nu$$:=$$w)^*$ \\
    $rassign ::= (\nu$$:=$$w)^*$ \\
    $syncs ::= \langle lock, \ell\rangle$\textbar$\langle unlock, \ell\rangle$\textbar$\langle signal, \gamma \rangle$\textbar$\langle wait, \gamma, \ell \rangle$ \\
\bottomrule
\end{tabular}
}
\end{center}
\end{table}

To express the operational semantics of a concurrent program for PDNet modeling, we define its labeled transition system (LTS) semantics based on Definition \ref{Def:CP} and Table \ref{Tab:syn}.

\begin{definition}[LTS Semantics of Concurrent Programs]\label{Def:LTS}
Given concurrent program $\mathcal{P}$, $\mathcal{N}_\mathcal{P}$ $::=$ $\langle \mathcal{S}, \mathcal{A}, \rightarrow\rangle$ is the labeled transition system of $\mathcal{P}$, where: 

1. $\mathcal{S}\subseteq\mathcal{H}\times\mathcal{M}\times (\mathcal{L}$$\rightarrow$$\mathcal{I})\times (\mathcal{C}$$\rightarrow $$\mathcal{I}_{MS})$ is a set of the program configurations.

2. $\mathcal{A}\subseteq\mathcal{T}\times\mathcal{B}$ is a set of actions, where $\mathcal{T}$ comes from $\mathcal{P}$.

3. $\rightarrow\subseteq\mathcal{S}\times\mathcal{A}\times\mathcal{S}$ is a set of transition relations on the program configurations $\mathcal{S}$.

\end{definition}

Formally, $s ::= \langle h, m, r, u\rangle$ is a configuration of $\mathcal{S}$, where $h$$\in$$\mathcal{H}$ is a function that indicates the current program location of every thread, $m$$\in$$\mathcal{M}$ is the current memory state, $r$ is a function that maps every mutex to a thread identifier, and $u$ is also a function maps every condition variable to a multiset of thread identifiers of those threads that currently wait on that condition variable.
$\mathcal{S}_0::=$$\langle h_0, m_0, r_0, u_0\rangle$ where $h_0$$\in$$\mathcal{H}$ and $m_0$$\in$$\mathcal{M}$ come from $\mathcal{P}$, $r_0$$:\mathcal{L}$$\rightarrow$$\{0\}$ represents that every mutex is not initially held by any threads, and $u_0$$:\mathcal{C}$$\rightarrow$$\emptyset$ represents that every condition variable does not initially block any threads.
Hence, we characterize the states of $\mathcal{P}$ by the configurations $\mathcal{S}$ of $\mathcal{N}_\mathcal{P}$.
$\alpha ::= \langle\tau, \beta\rangle$ is an action of $\mathcal{A}$, where $\tau$$\in$$\mathcal{T}$ is a statement of $\mathcal{P}$ and $\beta$$\in$$\mathcal{B}$ is an effect for operation $q$ from statement $\tau$.
The transition relation $\rightarrow$ on the configurations is represented by $s$$\stackrel{\langle\tau,\beta\rangle}{\longrightarrow}$$s'$.
The interleaved execution of $\tau$ could update configuration $s$ to new one $s'$ based on the effect $\beta$ corresponding to the operation of $\tau$.
In fact, the effect of an action $\alpha$$\in$$\mathcal{A}$ characterizes the nature of the transition relations with this action on configurations of $\mathcal{N}_\mathcal{P}$.
The effect is defined by $\mathcal{B}=(\{ass, jum, ret, tcd, fcd, call, rets\}$$\times$$K)\cup (\{acq, rel\}$$\times$$\mathcal{L})\cup (\{sig\}$$\times$$\mathcal{C})\cup(\{wa_1, wa_2, wa_3\}$$\times$$\mathcal{C}$$\times$$\mathcal{L}))$.

To formalize our PDNet modeling methods, the semantics of a concurrent program is expressed by the transition relations $\rightarrow$ on the program configurations under a current configuration $s$$=$$\langle h, m, r, u\rangle$ of $\mathcal{P}$ in Table \ref{Tab:sem}.
The intuition behind the semantics is how $s$ updates based on the transition relations with the actions of $\mathcal{A}$.
Thus, the execution of a statement gives rise to a transition relation in correspondence with the operation of the statement.
For convenience, the action referenced later is denoted by an abbreviation at the end of each row in Table \ref{Tab:sem}.
For instance, $ass$ represents the action $\langle \tau, \langle ass, l'\rangle\rangle$ where $\langle ass, l'\rangle$ is the effect corresponding to the operation of $\tau$, updating the program location to $l'$.
Here, suppose that an assignment operation is $\nu$$:=$$w$ in statement $\tau$.
$[\![w]\!]m$ denotes that the value evaluating by the expression $w$ under the memory state $m$. This value is assigned to variable $\nu$.
Thus, $m'$$=$$m[\nu$$\mapsto$$[\![w]\!]m]$ denotes a new memory state where $m'(\nu)$$=$$[\![w]\!]m$ and $m'(y)$$=$$m(y)$ ($\forall y$$\in$$\mathcal{V}$$:y$$\neq$$\nu$).

In the same way, $jum$ represents action $\langle \tau,\langle jum, l'\rangle\rangle$ where the jump operation of $\tau$ is $break$ or $continue$, updating the program location to $l'$. But $jum$ does not update a memory state.
$ret$ represents action $\langle \tau,\langle ret, l'\rangle\rangle$ where the jump operation of $\tau$ is $return$, updating the program location to $l'$. 
For a branching operation, $tcd$ represents action $\langle \tau,\langle tcd, l'\rangle\rangle$, where $[\![w]\!]m$$=$$true$, and $fcd$ represents action $\langle \tau,\langle tcd, l'\rangle\rangle$, where $[\![w]\!]m$$=$$flase$.
Neither $tcd$ nor $fcd$ updates the memory state. They update the program location to different one $l'$.
For a function call, $call$ represents action $\langle \tau, \langle call, l'\rangle\rangle$, where $l'$ is the entry of the called function, and $rets$ represents action $\langle \tau, \langle rets, l'\rangle\rangle$, where $l'$ is the return site of the calling function.
Similarly, suppose that $cassign$ of $call$ is $\nu_1$$:=$$w_1$ and $rassign$ of $rets$ is $\nu_2$$:=$$w_2$.
$m'$$=$$m[\nu_1$$\mapsto$$[\![w_1]\!]m]$ denotes the new memory state for $call$ and $m'$$=$$m[\nu_2$$\mapsto$$[\![w_2]\!]m]$ denotes the new one for $rets$.
In addition, $ass$, $jum$, $tcd$, $fcd$, $call$ and $rets$ do not update $r$ and $u$ of $s$.

\begin{table}[t]
\caption{Semantics of Concurrent Programs}\label{Tab:sem}
\begin{center}
\setlength{\tabcolsep}{0mm}{
\begin{tabular}{l}
\toprule
    $\frac{\tau := \langle i, q, l, l', m, m'\rangle\in\mathcal{T}\emph{ }q :=\nu:=w\emph{ }h(i)=l}
    {\langle h, m, r, u\rangle
    \stackrel{\langle \tau, \langle ass, l'\rangle\rangle}{\longrightarrow}
    \langle h[i\mapsto l'], m', r, u\rangle}$  ($ass$) \\

    $\frac{\tau := \langle i, q, l, l', m, m'\rangle\in\mathcal{T} q := break\emph{ }or \emph{ }continue\emph{ }h(i)=l}{\langle h, m, r, u\rangle\stackrel{\langle \tau, \langle jum, l'\rangle\rangle}{\longrightarrow}\langle h[i\mapsto l'], m, r, u\rangle}$  ($jum$) \\

    $\frac{\tau := \langle i, q, l, l', m, m'\rangle\in\mathcal{T}\emph{ }q := return\emph{ }h(i)=l}
    {\langle h, m, r, u\rangle
    \stackrel{\langle \tau, \langle ret, l'\rangle\rangle}{\longrightarrow}
    \langle h[i\mapsto l'], m', r, u\rangle}$  ($ret$) \\

    $\frac{\tau := \langle i, q, l, l', m, m'\rangle\in\mathcal{T}\emph{ }q := if(w)or while(w)\emph{ }[\![w]\!]m=true\emph{ }h(i)=l}
    {\langle h, m, r, u\rangle
    \stackrel{\langle \tau, \langle tcd, l'\rangle\rangle}{\longrightarrow}
    \langle h[i\mapsto l'], m, r, u\rangle}$ ($tcd$) \\

    $\frac{\tau := \langle i, q, l, l', m, m'\rangle\in\mathcal{T}\emph{ }q := if(w) or while(w) \emph{ } [\![w]\!]m=false \emph{ }h(i)=l}
    {\langle h, m, r, u\rangle
    \stackrel{\langle \tau, \langle fcd, l'\rangle\rangle}{\longrightarrow}
    \langle h[i\mapsto l'], m, r, u\rangle}$ ($fcd$) \\

    $\frac{\tau := \langle i, q, l, l', m, m'\rangle\in\mathcal{T}\emph{ }q := call\emph{ }h(i)=l}
    {\langle h, m, r, u\rangle
    \stackrel{\langle \tau, \langle call, l'\rangle\rangle}{\longrightarrow}
    \langle h[i\mapsto l'], m', r, u\rangle}$ ($call$) \\

    $\frac{\tau := \langle i, q, l, l', m, m'\rangle\in\mathcal{T}\emph{ }q := rets\emph{ }h(i)=l}
    {\langle h, m, r, u\rangle
    \stackrel{\langle \tau, \langle rets, l'\rangle\rangle}{\longrightarrow}
    \langle h[i\mapsto l'], m', r, u\rangle}$ ($rets$) \\

    $\frac{\tau := \langle i, q, l, l', m, m'\rangle\in\mathcal{T}\emph{ }q := \langle lock, \ell\rangle\emph{ }h(i)=l\emph{ }r(\ell)=0}
    {\langle h, m, r, u\rangle
    \stackrel{\langle \tau, \langle acq, \ell\rangle\rangle}{\longrightarrow}
    \langle h[i\mapsto l'], m, r[\ell\mapsto i], u\rangle}$ ($acq$) \\

    $\frac{\tau := \langle i, q, l, l', m, m'\rangle\in\mathcal{T}\emph{ }q := \langle unlock, \ell\rangle\emph{ }h(i)=l\emph{ }r(\ell)=i}
    {\langle h, m, r, u\rangle
    \stackrel{\langle \tau, \langle rel, \ell\rangle\rangle}{\longrightarrow}
    \langle h[i\mapsto l'], m, r[\ell\mapsto 0], u\rangle}$ ($rel$) \\

    $\frac{\tau := \langle i, q, l, l', m, m'\rangle\in\mathcal{T}\emph{ }q := \langle signal, \gamma\rangle\emph{ }h(i)=l\emph{ }\{j\}\in u(\gamma)}
    {\langle h, m, r, u\rangle
    \stackrel{\langle \tau, \langle sig, \gamma\rangle\rangle}{\longrightarrow}
    \langle h[i\mapsto l'], m, r, u[\gamma\mapsto u(\gamma)\setminus\{j\}\cup\{-j\}]\rangle}$ ($sig$) \\

    $\frac{\tau := \langle i, q, l, l', m, m'\rangle\in\mathcal{T}\emph{ }\emph{ }q := \langle wait, \gamma, \ell\rangle\emph{ }h(i)=l\emph{ }r(\ell)=i\emph{ }\{i\}\notin u(\gamma)}
    {\langle h, m, r, u\rangle
    \stackrel{\langle \tau, \langle wa_1, \gamma, \ell\rangle\rangle}{\longrightarrow}
    \langle h, m, r[\ell\mapsto 0], u[\gamma\mapsto u(\gamma)\cup\{i\}]\rangle}$ ($wa_1$) \\

    $\frac{\tau := \langle i, q, l, l', m, m'\rangle\in\mathcal{T}\emph{ }q := \langle wait, \gamma, \ell\rangle\emph{ }h(i)=l\emph{ }r(\ell)=0\emph{ }\{-i\}\in u(\gamma)}
    {\langle h, m, r, u\rangle
    \stackrel{\langle\tau,\langle wa_2,\gamma, \ell\rangle\rangle}{\longrightarrow}
    \langle h, m, r, u[\gamma\mapsto u(\gamma)\setminus\{-i\}]\rangle}$ ($wa_2$)
 \\

    $\frac{\tau := \langle i, q, l, l', m, m'\rangle\in\mathcal{T}\emph{ }q := \langle wait,\gamma,\ell\rangle\emph{ }h(i)=l\emph{ }r(\ell)=0}
    {\langle h, m, r, u\rangle
    \stackrel{\langle\tau,\langle wa_3,\gamma, \ell\rangle\rangle}{\longrightarrow}
    \langle h[i\mapsto l'], m, r[\ell\mapsto i], u\rangle}$($wa_3$) \\
\bottomrule
\end{tabular}}
\end{center}
\end{table}

Moreover, $acq$ represents action $\langle \tau,\langle acq,\ell\rangle$ corresponding to operation $\langle lock, \ell\rangle$. If $\ell$ is not held by any thread ($r(\ell)$$=$$0$), $acq$ represents that thread $i$ obtains this mutex $\ell$ ($r[\ell$$\mapsto$$i]$) and updates the program location to $l'$.
However, if $\ell$ is held by another thread, thread $i$ could be blocked by $\ell$, and current configuration $s$ cannot be updated by $acq$. 
$rel$ represents action $\langle \tau,\langle rel,\ell\rangle$ corresponding to operation $\langle unlock, \ell\rangle$.
Here, $r(\ell)$$=$$i$ means that the mutex $\ell$ is held by thread $i$. If $r(\ell)$$=$$i$, thread $i$ could release this mutex $\ell$ ($r[\ell$$\mapsto$$0]$) and updates the program location to $l'$.
Then, $sig$ represents action $\langle \tau, \langle sig, \gamma\rangle\rangle$ corresponding to operation $\langle signal, \gamma\rangle$. Thread $i$ could notify thread $j$ belonging to $u(\gamma)$ ($\{j\}$$\in$$u(\gamma)$). Thus, thread $j$ could be notified by thread $i$ ($u[\gamma$$\mapsto$$u(\gamma)$$\setminus$$\{j\}$$\cup$$\{-j\}]$).
It updates the program location to $l'$.
Particularly, operation $\langle wait, \gamma, \ell\rangle$ corresponds to three actions $wa_1$, $wa_2$ and $wa_3$, where only $wa_3$ updates the program location to $l'$.
If mutex $\ell$ is held by thread $i$ ($r(\ell)$$=$$i$) and thread $i$ is not waiting for $\gamma$ currently ($\{i\}$$\notin$$u(\gamma)$), $wa_1$ ($\langle wa_1, \gamma,\ell\rangle$) represents that the action releases mutex $\ell$ ($r[\ell$$\mapsto$$0]$), and thread $i$ is added to the current thread multiset waiting on condition variable $\gamma$ ($u[\gamma$$\mapsto$$u(\gamma)$$\cup$$\{i\}]$).
Then, $wa_2$ ($\langle wa_2, \gamma,\ell\rangle$) represents that the thread $i$ is blocked until a thread $j$ ($\{-i\}$$\in$$u(\gamma)$) is notified by condition variable $\gamma$ . Thus, thread $i$ no long waits for a notification on $\gamma$ ($u[\gamma$$\mapsto$$u(\gamma)$$\setminus$$\{-i\}]$).
Finally, if $\ell$ is not held ($r(\ell)=0$), $wa_3$ ($\langle wa_3, \gamma,\ell\rangle$) represents that the action acquires mutex $\ell$ again ($r[\ell$$\mapsto$$i]$) and updates the program location to $i$.
In addition, $acq$, $rel$, $sig$, $wa_1$, $wa_2$ and $wa_3$ do not update $m$ of the current configuration $s$.
\section{Post-process Algorithm}\label{App:Post}

\begin{breakablealgorithm}\caption{Slicing Post-process Algorithm}\label{Alg:Post}
\begin{algorithmic}[1]
\Require A PDNet $N$ and its slice $N'$ w.r.t. $Crit$;
\Ensure The final PDNet slice $N''$;
    \ForAll {$p$$\in$$(P'$$\cap$$P_f)\wedge(\forall t$$\in$$^\bullet$$p$$\colon(t,p)$$\notin$$F_f'$)}
    /$\ast F_f'$ is the execution arc set$\ast$/
       \State $T_f :=$ \Call {FINDEXE}{$t,T$}; /$\ast T$ is the transition set of $N\ast$/
    \ForAll {$t'$$\in$$T_f$}
       \State ADDARC($t',p$); /$\ast (t',p)$ is a new execution arc$\ast$/
    \EndFor
\EndFor
\end{algorithmic}
\end{breakablealgorithm}

If there exists an execution place $p$ of $N'$ that lacks execution arcs connected to any transition in $^\bullet p$, the execution order of $p$ is not complete (Line 1). Function FINDEXE($t,T$) can find the transitions performed before the transition in $^\bullet$$p$ based on $T$ (Line 2). For each transition $t'$ in $T_f$, a new execution arc $(t',p)$ is constructed by ADDARC($t',p$) (Line 4) to complete a control-flow structure. Finally, the final PDNet slice $N''$ is tractable for the model checking.

\section{Correctness Proof}\label{App:Ana}

We prove the correctness of PDNet slice $N'$ expressed by $N\models\psi$$\Leftrightarrow$$N'\models\psi$.
Firstly, we give some definitions for the correctness proof.
Note that an occurrence sequence is defined in Definition \ref{Def:Seq}, and an LTL-$_\mathcal{X}$ formula is defined in Definition \ref{Def:LTL}.

\begin{definition}[Marking Projection]\label{Def:proj}
Let $N$ be a PDNet, $M$ a marking of $N$, $\psi$ an LTL-$_\mathcal{X}$ formula of $N$, and $Crit$ a place subset extracted from $\psi$.
A projection function $\downarrow[\psi]M:Crit$$\rightarrow$$\mathbb{E}_\emptyset$ is a local marking function w.r.t. $Crit$, that assigns an expression $M(p)$ to each place $p$. $\forall p$$\in$$Crit$$:$ $Type[M(p)]=C(p)_{MS}\wedge(Var(M(p))=\emptyset)$.
\end{definition}

\begin{definition}[$\psi$-equivalent Marking]\label{Def:equM}
Let $N$ be a PDNet, $\psi$ an LTL-$_\mathcal{X}$ formula of $N$, and $M$ and $M'$ two markings of $N$.
$M$ and $M'$ are $\psi$-equivalent, denoted by $M\stackrel{\psi}{\rightsquigarrow}M'$, if $\downarrow[\psi]M=\downarrow[\psi]M'$.
\end{definition}

\begin{definition}[$\psi$-stuttering Equivalent Transition] \label{Def:equT}
Let $N$ be a PDNet, $\psi$ an LTL-$_\mathcal{X}$ formula of $N$, $\omega$ the occurrence sequence $(t_1, b_1)$, $(t_2, b_2)$, $\cdots$, $(t_{i-1}, b_{i-1})$, $(t_i, b_i)$, $(t_{i+1}, b_{i+1})$, $\cdots$, $(t_n, b_n)$ of $N$ defined in Definition \ref{Def:Seq}, and $M(\omega)$ the marking sequence $M_0$, $M_1$, $\cdots$, $M_{i-1}$, $M_i$, $M_{i+1}$, $\cdots$, $M_n$ generated by occurring every binding element of $\omega$ in turn.
$t_i$ ($1$$\leq$$i$$\leq$$n$) is a $\psi$-stuttering equivalent transition if $M_{i-1}\stackrel{\psi}{\rightsquigarrow}M_{i}$.
\end{definition}

\begin{definition}[$\psi$-stuttering Equivalent Marking Sequences] \label{Def:equMS}
Let $N$ be a PDNet, $\psi$ an LTL-$_\mathcal{X}$ formula of $N$, $\omega$ an occurrence sequence of $N$, $M(\omega)$ the marking sequence of $\omega$ starting from $M_0$, and $\omega'$ be generated by eliminating some binding elements from $\omega$, and $M(\omega')$ marking sequence of $\omega'$ starting from $M_0'$.
$M(\omega)$ and $M(\omega')$ are $\psi$-stuttering equivalent marking sequences, denoted by $M(\omega)\stackrel{st_\psi}{\rightsquigarrow}M(\omega')$, if
$\downarrow[\psi]M_0=\downarrow[\psi]M_0'$, and
the eliminated transitions from $\omega$ are $\psi$-stuttering equivalent transitions.
\end{definition}

Intuitively, the $\psi$-non-stuttering equivalent transitions of $\omega$ are all preserved in $\omega'$.

\begin{definition}[$\psi$-stuttering Equivalent PDNets] \label{Def:equP}
Let $N$ be a PDNet, $\psi$ an LTL-$_\mathcal{X}$ formula of $N$, $N'$ a PDNet with the same LTL-$_\mathcal{X}$ formula $\psi$, $M_0$ the initial marking of $N$, and $M_0'$ the initial marking of $N'$.
$N$ and $N'$ are $\psi$-stuttering equivalent PDNets, denoted by $N\stackrel{st_\psi}{\rightsquigarrow}N'$, if

1) $\downarrow[\psi]M_0=\downarrow[\psi]M_0'$, and

2) for each marking sequence $M(\omega)$ of $N$ starting from $M_0$, there exists a marking sequence $M(\omega')$ of $N'$ starting from $M_0'$ such that $M(\omega)\stackrel{st_\psi}{\rightsquigarrow}M(\omega')$, and

3) for each marking sequence $M(\omega')$ of $N'$ starting from $M_0'$, there exists a marking sequence $M(\omega)$ of $N$ starting from $M_0$ such that $M(\omega')\stackrel{st_\psi}{\rightsquigarrow}M(\omega)$.
\end{definition}

\begin{theorem}\label{The:stuInv}
\cite{Peled1997Stutter}
Let $N$ be a PDNet, $\psi$ an LTL-$_\mathcal{X}$ formula of $N$, and $N'$ a PDNet with the same LTL-$_\mathcal{X}$ formula $\psi$.
An LTL-$_\mathcal{X}$ formula is invariant under stuttering, denoted by $N\stackrel{st_\psi}{\rightsquigarrow}N'$, iff $N\models\psi\Leftrightarrow N'\models\psi$.
\end{theorem}

The above theorem shows that an LTL-$_\mathcal{X}$ formula $\psi$ is invariant under stuttering \cite{Peled1997Stutter}.

\begin{theorem}\label{The:stuInvPDNet}
Let $N$ be a PDNet, $M_0$ the initial marking of $N$, $\psi$ an LTL-$_\mathcal{X}$ formula of $N$, $Crit$ be extracted from $\psi$, $N'$ w.r.t. $Crit$ the final PDNet slice from $N$ by Algorithm \ref{Alg:Slicing} and \ref{Alg:Post}, and $M_0'$ the initial marking of $N'$.
$N\stackrel{st_\psi}{\rightsquigarrow}N'$.
\end{theorem}

\begin{proof}
According to Definition \ref{Def:equP}'s 1), $\downarrow[\psi]M_0=\downarrow[\psi]M_0'$ holds, because we keep all places in $Crit$, and Algorithm \ref{Alg:Slicing} does not change the initial marking of $N$.

Then, we prove Definition \ref{Def:equP}'s 2).
Let $\omega$ be an arbitrary occurrence sequence of $N$ arbitrarily. There exists an occurrence sequence $\omega'$ that is generated by eliminating some binding elements from $\omega$ by Algorithm \ref{Alg:Slicing}.
$M(\omega)$ is the marking sequence corresponding to $\omega$, and $M(\omega')$ is the marking sequence corresponding to $\omega'$.
We prove $M(\omega)\stackrel{st_\psi}{\rightsquigarrow}M(\omega')$ by using structural induction on the length of occurrence sequence of $\omega$, denoted by $\vert\omega\vert$.

 \textit{Base case}: Let $\vert\omega\vert$$=$$1$. $\downarrow[\psi]M_0=\downarrow[\psi]M_0'$ holds.

 \textit{Induction}: Assume that it holds when \textbar$\omega$\textbar$=$$k$. That is, $\omega$$=$$(t_{i1}, b_{i1})$, $(t_{i2}, b_{i2})$, $\cdots$, $(t_{ik}, b_{ik})$ of $N$, $\omega'$$=$$(t_{j1}, b_{j1})$, $(t_{j2}, b_{j2})$, $\cdots$, $(t_{jm}, b_{jm})$ ($m$$\leq$$k$), and $M(\omega)\stackrel{st_\psi}{\rightsquigarrow}M(\omega')$.
 Then we consider whether it holds when $\vert\omega\vert$$=$$k$$+$$1$.
 Suppose that the last transition is $t_{k+1}$, and $\omega(t_{k+1},b_{k+1})$ is the extended occurrence sequence.
 There are two cases for $\omega'$: $t_{k+1}$ is sliced and $t_{k+1}$ is not sliced.

 Case 1: $t_{k+1}$ is sliced by Algorithm \ref{Alg:Slicing}.
 That is, ${\nexists}p$$\in$$Crit$$:$ $t_{k+1}$$\stackrel{d}{\longrightarrow}$$^*$$p$.
 Consider the two proposition forms. Let $po$ be a proposition from $\psi$ arbitrarily.
 If $po$ is in the form of $token-value(p_t)\emph{ }rop\emph{ }c$, $t_{k+1}$$\notin$$^\bullet p_t$ or $E(t_{k+1},p_t)$$=$$E(p_t,t_{k+1})$, and the evaluation of this proposition under each marking is not change.
 If $po$ is in the form of $is\mbox{-} fireable(t)$, $t_{k+1}$ does not affect the enabling condition of $t$ according to Definition \ref{Def:Crit} and Algorithm \ref{Alg:Slicing}, and the evaluation of this proposition under each marking is not change.
 Thus, $t_{k+1}$ is a $\psi$-stuttering equivalent transition. As a result, $M(\omega(t_{k+1},b_{k+1}))\stackrel{st_\psi}{\rightsquigarrow}M(\omega')$.

 Case 2: $t_{k+1}$ is not sliced, and $\omega'(t_{k+1},b_{k+1})$ is the extended occurrence sequence.
 According to $M(\omega)\stackrel{st_\psi}{\rightsquigarrow}M(\omega')$, $t_{k+1}$ produces the same effect for $\omega$ and $\omega'$.
 As a result, $M(\omega(t_{k+1},b_{k+1}))\stackrel{st_\psi}{\rightsquigarrow}M(\omega'(t_{k+1},b_{k+1}))$.

 Finally, we prove Definition \ref{Def:equP}'s 3).
 Let $\omega'$ be an arbitrary occurrence sequence of $N'$. There exists an occurrence sequence $\omega$ that is generated by adding some binding elements to $\omega'$ that are identified by Algorithm \ref{Alg:Slicing}.
 $M(\omega')$ is the marking sequence corresponding to $\omega'$, and $M(\omega)$ is the marking sequence corresponding to $\omega$.
 Then, we prove that $M(\omega')\stackrel{st_\psi}{\rightsquigarrow}M(\omega)$ by using structural induction on the length of occurrence sequence $\omega'$, denoted by $\vert\omega'\vert$.

 \textit{Base case}: Let \textbar$\omega'$\textbar$=$$1$. $\downarrow[\psi]M_0'=\downarrow[\psi]M_0$ holds.

 \textit{Induction}: Assume that it holds when \textbar$\omega'$\textbar$=$$k'$.
 That is, $\omega'$$=$$(t_{i1}, b_{i1})$, $(t_{i2}, b_{i2})$, $\cdots$, $(t_{ik'}, b_{ik'})$ of $N$, $\omega$$=$$(t_{j1}, b_{j1})$, $(t_{j2}, b_{j2})$, $\cdots$, $(t_{jm'}, b_{jm'})$ ($m'$$\geq$$k'$), and $M(\omega')\stackrel{st_\psi}{\rightsquigarrow}M(\omega)$.
 Then we consider whether it holds when \textbar$\omega'$\textbar$=$$k'$$+$$1$.
 Suppose the last transition is $t_{k'+1}$, and $\omega'(t_{k'+1},b_{k'+1})$ is the extended occurrence sequence.
 Obviously, $t_{k'+1}$ belongs to $\omega$, and $\omega(t_{k'+1},b_{k'+1})$ is the extended occurrence sequence.
 $t_{k'+1}$ produces the same effect for $\omega$ and $\omega'$. As a result, $M(\omega'(t_{k'+1},b_{k'+1}))\stackrel{st_\psi}{\rightsquigarrow}M(\omega(t_{k'+1},b_{k'+1}))$.

 Thus, $N\stackrel{st_\psi}{\rightsquigarrow}N'$ holds based on Definition \ref{Def:equP}.
\end{proof}

\begin{theorem}\label{Cor:sas}
Let $N$ be a PDNet, $\psi$ an LTL-$_\mathcal{X}$ formula, $Crit$ be extracted from $\psi$, and $N'$ w.r.t. $Crit$ be a reduced PDNet by Algorithm \ref{Alg:Slicing}.
$N\models\psi\Leftrightarrow N'\models\psi$. 
\end{theorem}

\begin{proof}
It is obvious that $N\stackrel{st_\psi}{\rightsquigarrow}N'$ proved by Theorem \ref{The:stuInvPDNet}. Thus, $N\models\psi\Leftrightarrow N'\models\psi$ iff $N\stackrel{st_\psi}{\rightsquigarrow}N'$ according to Theorem \ref{The:stuInv}. Therefore, this corollary holds.
\end{proof}

Hence, it is concluded that the PDNet slice obtained by our methods is correct based on Theorem \ref{Cor:sas}.

\bio{}
\endbio

\bio{}
\endbio

\end{document}